\documentclass[aps,pra,twocolumn,showpacs,superscriptaddress]{revtex4-2} 
\pdfoutput=1

\usepackage{graphicx}
\usepackage{verbatim}
\usepackage{physics}
\usepackage{bbold}
\usepackage{mathtools}
\usepackage{amsmath, amsthm, amssymb,mathrsfs}

\newtheorem{proposition}{Proposition}
\newtheorem*{proposition*}{Proposition}
\usepackage{breqn}
\usepackage[colorlinks]{hyperref}

\newcommand\scalemath[2]{\scalebox{#1}{\mbox{\ensuremath{\displaystyle #2}}}}

\usepackage{xcolor}
\usepackage{hyperref}

\date{\today}

\makeatletter
\let\cat@comma@active\@empty
\makeatother

\begin{document}
\title{Noncommutative polynomial optimization under symmetry}
\author{Marie Ioannou}
\address{Department of Applied Physics University of Geneva, 1211 Geneva, Switzerland}
\author{Denis Rosset}
\address{Department of Applied Physics University of Geneva, 1211 Geneva, Switzerland}
\begin{abstract}
    We present a general framework to exploit the symmetries present in the Navascu{\'e}s-Pironio-Ac{\'i}n semidefinite relaxations that approximate invariant noncommutative polynomial optimization problems.
    We put equal emphasis on the moment and sum-of-squares dual approaches, and provide a pedagogical and formal introduction to the Navascu{\'e}s-Pironio-Ac{\'i}n technique before working out the impact of symmetries present in the problem.
    Using our formalism, we compute analytical sum-of-square certificates for various Bell inequalities, and prove a long-standing conjecture about the exact maximal quantum violation of the CGLMP inequalities for dimension 3 and 4.
    We also apply our technique to the Sliwa inequalities in the Bell scenario with three parties with binary measurements settings/outcomes.
    Symmetry reduction is key to scale the applications of the NPA relaxation, and our formalism encompasses and generalizes the approaches found in the literature.
\end{abstract}

\maketitle

\section{Introduction}

Noncommutative polynomial optimization is a widely used tool in quantum information.
Many applications of interest involve a quantum state on which parties apply operators.
In the Bell tests in quantum nonlocality, the parties try to select an optimal quantum state and measurement operators to maximize the violation of a Bell inequality.
It may also be the adversary in a quantum key distribution task that tries to maximize the information it can guess about the secret being shared.
In these cases, the initial state of the system, as well as the operators acting on it are variables.
It is possible to encode the structure of the problem into an abstract algebra, so that the optimization problem reduces to a noncommutative polynomial optimization problem subject to constraints.
In general, this problem cannot be solved directly.
Lower bounds are obtained, in general, by fixing the dimension to a reasonable number so that explicit realizations can be parameterized and optimized over using random starting points and local solvers~\cite{Pal_2009}.
On the other hand, upper bounds are obtained using the hierarchy of semidefinite~\cite{Vandenberghe_1996,boyd2004convex} relaxations proposed by Navascu{\'e}s, Pironio, Ac{\'i}n (NPA)~\cite{Navascues2007,Navascues_2008,Pironio_2010}.
One formulation of the NPA hierarchy provides an optimality certificate.
The soundness of this certificate can be verified independently without solving the semidefinite problem again.
The NPA hierarchy has been used to find optimal quantum bounds, and to certify properties of quantum devices through the self-testing approach~\cite{Bamps_2015,Supic_2016,Salavrakos_2017,Augusiak_2019,Kaniewski_2019}, see~\cite{Supic2020} for a review, and~\cite{Mayers_2004} for an application to quantum cryptography~\cite{Mayers_1998}.
In this approach, the certification is done in a device-independent way, and guarantees that the devices have no malfunction nor tampering; such methods are of great interest to secure communications.

The noncommutative polynomial optimization problem has a commutative counterpart; the corresponding semidefinite programming relaxations are known as the Lasserre hierarchy~\cite{Lasserre2001,Parrilo2003a}.
This commutative hierarchy has been used in quantum information as well.
One of the first examples is the computation of bounds on Bell inequalities for states of a given dimension~\cite{Liang2007a}.
Later, the technique was used in many body quantum systems~\cite{Fadel2017,Baccari2017} to characterize the classical/quantum boundary; also to bound the sets of classical correlations in networks~\cite{Pozas-Kerstjens2019}.

Most of the scenarios listed above exhibit some degree of symmetry.
Parrilo and Gatermann~\cite{Gatermann_2004} proposed a variant of the commutative hierarchy that exploits the symmetries present in the problem.
In quantum information/noncommutative setting, several papers exploited symmetries to reduce the problem size~\cite{Bamps_2015,Cai2016,Aguilar2018,Rosset2018,Tavakoli2019,Agresti2021,Tavakoli2021}.
When the problem is small, the symmetry reduction can be sufficient to bring the problem to a form that is solvable analytically.
In larger instances, symmetry reduction is often necessary to solve the problem at all numerically.
If we consider CPU time, it is possible to leave a solver running for a few weeks; however, the available memory is a hard limit that prevents the computation altogether.
Complete symmetrization as described in this paper includes block-diagonalization and helps for both.

Symmetry reduction can be tricky to get right, especially given the rich structure already present in the NPA hierarchy.
While a formal description of it was given in a mathematics journal~\cite{Pironio_2010}, the same level of detail is not always present in quantum information publications.
A second difficulty appears when using the Collins-Gisin enumeration of POVM elements:
permuting measurement outcomes is no longer a unitary operation.
Most textbooks about groups and representations assume unitarity right at the beginning, including~\cite{Serre} and the Gatermann and Parrilo commutative hierarchy~\cite{Gatermann_2004}.

To alleviate those difficulties, we propose in this paper a self-contained description of the NPA hierarchy and its invariance properties under symmetry.
In particular, we consider on equal footing the moment-based and sum-of-squares (SOS) constructions -- as these represent dual points of view of the same mathematical object.
We revisit the formal description of~\cite{Pironio_2010}, in particular the description of polynomial rings under equality constraints, but we tailor our examples from the study of Bell inequalities.
In our pedagogical exposition, we remove the assumptions that are done for the sake of convenience: we allow problems with complex coefficients and with symmetries that are not necessarily unitary.
For clarity, we focus on small scale problems where symmetrization enables the computation of exact results.
Using exact arithmetic, we compute and verify the quantum bound of the CGLMP~\cite{Collins_2002} inequality for $d=2,3,4$.
For $d=2$, the inequality reduces to the well-known CHSH~\cite{Clauser1974} inequality, for which an optimal certificate has been derived a number of times~\cite{Reichardt2013,Bamps_2015}.
For $d=3,4$, we prove for the first time the optimality of the bound conjectured in~\cite{Acin_2002} and verified up to numerical precision in~\cite{Navascues2007,Navascues_2008}, thus solving partially the 27th open problem from Werner's list~\cite{Krueger2005}.
We also compute the almost quantum bound for four of Sliwa's inequalities~\cite{Sliwa_2003,Vallins2017} in the Bell scenario with three parties, and binary settings/outcomes.

Our paper is structured as follows.

The first sections present the NPA hierarchy.
In Section~\ref{sec:problemtosolve} we give an informal introduction to the problem we aim to solve.
In Section~\ref{sec:definitions}, we define formally the polynomials used in our optimization problems.
In Section~\ref{sec:semidefiniterelaxations}, we provide a self-contained introduction to the NPA hierarchy, first through moment-based relaxations, then using the sum-of-squares approach, and then comment on their duality.
We then discuss the extraction of certificates, so that exact bounds can be verified independently of the use of a numerical solver.
We conclude the section with the construction of a SOS certificate for the CHSH inequality.

The subsequent sections discuss the symmetries.
In Section~\ref{sec:symmetriesdefs}, we first consider the complex conjugate invariance as a symmetry, and demonstrate the usefulness of convexity when symmetries are present.
We then put the symmetries of polynomial rings and optimization problems on a formal footing.
Those symmetries will impact the semidefinite relaxations, thus we derive the block-diagonal form adequate for Hermitian matrices invariant under possibly non-unitary transformations.
We conclude the section by discussing computational challenges.
In Section~\ref{sec:symadaptedsdp}, we study the impact of those symmetries on the moment-based and sum-of-squares approaches to the NPA hierarchy.
We formulate their block-diagonalization so that semidefinite duality is preserved for the symmetrized problems, and show how to recover SOS certificates.
For pedagogical purposes, we interleave the theory with the symmetrization of the CHSH example already studied.

The final sections of our paper present applications.
We investigate the properties of the CGLMP inequality in Section~\ref{sec:cglmpgeneralities}, including its invariance properties, the irreducible representations of its symmetry group and its conjectured maximal quantum violation for $d=3,4,5$ in exact arithmetic.
Interestingly, while the optimal quantum measurements have a regular structure, we observe that the complexity of the optimal state coefficients grows with $d$.
In Section~\ref{sec:cglmpcomputations}, we describe a step-by-step approach to derive SOS certificates, and the prove the optimality of the conjectured bounds for $d=3$ and $d=4$.
We finally turn to Sliwa's inequalities in Section~\ref{sec:Sliwa}, where we provide four SOS certificates and list the symmetry groups of the 46 Bell inequality families present in the scenario.

\section{Problem to be solved}
\label{sec:problemtosolve}
We consider a generic bipartite Bell expression of the type
\begin{equation}
  E^{\text{Bell}} = \sum_{a b x y} c_{a b x y} P (a b | x y).
\end{equation}
The indices $x = 0,...,n_X-1$ and $y=0,...,n_Y-1$ are settings, while the indices $a = 0,...,n_A-1$ and $b=0,...,n_B-1$ are outcomes of measurements done on distinct subsystems of a quantum state. 
This Bell expression is possibly associated with a local bound $\mu_{\mathcal{L}}$ so that $E^{\text{Bell}} \leqslant\mu_{\mathcal{L}}$ for all local correlations.

We now consider the problem of finding the quantum supremum $\mu_{\mathcal{Q}}$ so that $E^{\text{Bell}}\leqslant \mu_{\mathcal{Q}}$ for all quantum systems. 
We use the commuting-operator formalism ~\cite{Slofstra_2020} and write $P (a b | x y) = \langle \psi | \mathcal{A}_{a | x} \mathcal{B}_{b | y} | \psi \rangle$, where the measured $|\psi \rangle \in \mathcal{H}$ is a normalized quantum state in an arbitrary Hilbert space $\mathcal{H}$. 
The measurements are described on the first and second subsystem by the respective Hermitian measurement operators $\mathcal{A}_{a|x} : \mathcal{H} \rightarrow \mathcal{H}$ and $\mathcal{B}_{b|y} : \mathcal{H} \rightarrow \mathcal{H}$. 
By Naimark's dilation theorem, we can assume that those operators are projectors, so that, for each $x$, $\{ \mathcal{A}_{a|x} \}_a$ is {\em{projection-valued measure}} (PVM), and the same for $\mathcal{B}_{b|y}$.
The state obeys the constraint:
\begin{equation}
    \label{eq:normalization}
    \langle \psi | \psi \rangle = 1,
\end{equation}
while the operators obey the relations:
\begin{equation}
  \label{eq:completeness} 1 - \sum_a \mathcal{A}_{a|x} = 1 - \sum_b
  \mathcal{B}_{b|y} = 0,
\end{equation}
\begin{equation}
  \label{eq:commutation} 
  \mathcal{A}_{a| x} \mathcal{B}_{b|y}
  -\mathcal{B}_{b|y} \mathcal{A}_{a| x} = 0,
\end{equation}
\begin{equation}
  \label{eq:projector} 
  \begin{array}{llll}
    (\mathcal{A}_{a|x})^2 -\mathcal{A}_{a|x} = 0, & & \mathcal{A}_{a|x} \mathcal{A}_{a'|x} = 0 & (a \neq a'),\\
    (\mathcal{B}_{b|y})^2 -\mathcal{B}_{b|y} = 0, & & \mathcal{B}_{b|y} \mathcal{B}_{b'|y} = 0 & (b \neq b'),
  \end{array}
\end{equation} over combinations of indices $a, a', b, b', x, y$. Eq.~\eqref{eq:commutation} encodes the commutation of measurements, while Eqs.~\eqref{eq:completeness} and~\eqref{eq:projector} describe projective measurements. 
In particular, Eq.~\eqref{eq:projector} implies that all operators are positive semidefinite.
Eq.~\eqref{eq:completeness} enables us to remove all operators with either $a = n_A-1$ or $b = n_B-1$ by defining:
\begin{equation}
    \label{eq:removelastoutcome}
    \mathcal{A}_{n_A-1|x} = 1 - \sum_{a=0}^{n_A-2} \mathcal{A}_{a|x}, \quad \mathcal{B}_{n_B-1|y} = 1 - \sum_{b=0}^{n_B-2} \mathcal{B}_{b|y}
\end{equation}
for all $x$ and $y$.
We rewrite the objective function
\begin{equation}
\label{eq:objectivefunction}
  \mathcal{E}^{\text{Bell}} = \delta + \sum_{a x} \alpha_{a x} \mathcal{A}_{a | x} + \sum_{b y} \beta_{b y} \mathcal{B}_{b | y} + \sum_{a b x y} \gamma_{a b x y} \mathcal{A}_{a | x} \mathcal{B}_{b | y},
\end{equation}
where in particular $a \le n_A-2$ and $b \le n_B-2$. 
The expectation value of the objective function $\mathcal{E}^{\text{Bell}}$ is then $E^{\text{Bell}} = \langle \psi | \mathcal{E}^{\text{Bell}} | \psi \rangle$ in Dirac notation with $\ket{\psi}$ a normalized state vector Eq.~\eqref{eq:normalization}.

Let us solve the optimization problem $\mu_{\mathcal{Q}} = \sup \mathcal{E}^{\text{Bell}}$:
\begin{equation}
\label{eq:optimizationproblem}
    \mu_{\mathcal{Q}} = \sup_{\ket{\psi}, \{\mathcal{A}_{a|x}\}_{a,x},\{\mathcal{B}_{b|y}\}_{b,y}} \sum_{abxy} c_{abxy}\bra{\psi}\mathcal{A}_{a|x} \mathcal{B}_{b|y}\ket{\psi},
\end{equation}
subject to the constraints~\eqref{eq:normalization}-\eqref{eq:projector}.

On can say that the bound is {\em tight} when $E^{\text{Bell}} = \mu_{\mathcal{Q}}$ for particular quantum correlations. 
In this paper, we will compute a quantum upper bound for $E^{\text{Bell}}$ in the framework of noncommutative polynomial optimization. 
In Ref.~\cite{Navascues_2008}, the problem is solved by an approximation consisting of a hierachy such that the quantum set $\mathcal{Q}$ is outer approximated. 
The approximation uses semidefinite programming and $\mathcal{Q}_1 \supseteq \mathcal{Q}_2 \supseteq ... \supseteq \mathcal{Q}_n$ converges to $\mathcal{Q}$ as $n\rightarrow\infty$. 
Until now, there is no analytical characterization of the quantum set $\mathcal{Q}$ even though it is convex. 
Therefore, finding the optimal quantum violation of a Bell inequality is in general a hard task.

\section{Definitions}
\label{sec:definitions}
Let us define formally the space that $\mathcal{E}^{\text{Bell}}$ inhabits. 
In the next subsection, we strive to be formal and succinct; nevertheless there should be no surprises, as the content matches how in practice people compute with noncommutative polynomials in the quantum information literature.

\subsection{Noncommutative polynomial rings}

We write $\mathcal{P}_{\text{free}} = \mathbb{C} [\mathcal{X}_1, \ldots,\mathcal{X}_N]$ the free noncommutative polynomial ring in $N$ variables over the complex numbers. 
Let $\overline{i} = (i_1, \ldots, i_d)$ be a sequence of integers between $1$ and $N$. 
We write
\begin{equation}
\label{eq:monomial}
  \mathcal{M}_{\overline{i}} =\mathcal{X}_{i_1} \mathcal{X}_{i_2} \ldots
  \mathcal{X}_{i_d}
\end{equation}
a monomial of degree $d = \deg \left( \mathcal{M}_{\overline{i}} \right)$ corresponding to the sequence $\overline{i}$. 
The multiplication of monomials is given by concatenation of the variables in Eq.~\eqref{eq:monomial}.
An element $\mathcal{S} \in \mathcal{P}_{\text{free}}$ is a formal sum
\begin{equation}
  \mathcal{S}= \sum_{\overline{i}} \mu_{\overline{i}}
  \mathcal{M}_{\overline{i}}
\end{equation}
where the sum is finite and runs over all sequences $\overline{i}$ such that $\mu_{\overline{i}} \in \mathbb{C}$ is nonzero. 
The degree of $\mathcal{S}$ is
\begin{equation}
  \deg \mathcal{S}= \max_{\overline{i}} \deg \left( \mathcal{M}_{\overline{i}}
  \right) .
\end{equation}
Multiplication on $\mathcal{P}_{\text{free}}$ follows from the concatenation of monomials. 
We equip $\mathcal{P}_{\text{free}}$ with a linear involution $\ast$ such that
\begin{equation}
  (\mathcal{X}_{i_1} \mathcal{X}_{i_2} \ldots \mathcal{X}_{i_d})^\ast =
  (\mathcal{X}_{i_d} \ldots \mathcal{X}_{i_2} \mathcal{X}_{i_1}).
\end{equation}
Hence, $\mathcal{P}_{\text{free}}$ is formally a $\ast$-algebra~\footnote{The definition assumes $\mathcal{X}_i^\ast = \mathcal{X}_i$ for all $i$. Indeed, we assume self-adjoint operators to alleviate the notation burden, but this restriction is easily relaxed for problems more general than optimizing over Bell inequalities.}.

Going back to our quantum optimization problem, we enumerate all operators $(\mathcal{A}_{0|0}, \ldots, \mathcal{A}_{n_A-2|0}, \mathcal{A}_{0|1}, \ldots, \mathcal{B}_{0|0}, \ldots, \mathcal{B}_{n_B-2|0}, \mathcal{B}_{0|1}, \ldots)$ and put them in correspondence with the list $(\mathcal{X}_1,\ldots, \mathcal{X}_N)$ for ease of notation.
We note that we excluded from the variables the operators with either $a=n_A-1$ or $b=n_B-1$ using Eq.~\eqref{eq:removelastoutcome}. 

We deal with the remaining equality constraints~\eqref{eq:commutation} and~\eqref{eq:projector} using the approach discussed in Sec. 3.5. of~\cite{Pironio_2010}.
We put the left-hand side of the remaining relations in a finite set $\overline{\mathcal{R}} \subseteq \mathcal{P}_{\text{free}}$, such that $\mathcal{R} = 0$ for all $\mathcal{R} \in \overline{\mathcal{R}}$.
It is clear that $\left \langle \psi \middle | \mathcal{R} \middle | \psi \right \rangle = 0$, but $\left \langle \psi \middle | \mathcal{S} \mathcal{R} \mathcal{T} \middle | \psi \right \rangle = 0$ holds true for arbitrary polynomials $\mathcal{S}, \mathcal{T}$ as well.
We write $\mathcal{I} = \left\langle \overline{\mathcal{R}} \right\rangle$ the set of all expressions equal to $0$ by a similar reasoning.
It is a set closed under addition, and under left/right multiplication by elements of $\mathcal{P}_{\text{free}}$, and thus a {\em two-sided ideal} of that polynomial ring~\cite{Mora_1994}.
Formally:
\begin{eqnarray}
  \mathcal{I} & = & \left\langle \overline{\mathcal{R}} \right\rangle \\
  & := & \left \{ \mathcal{S}_1 \mathcal{R}_1 \mathcal{T}_1 + \ldots + \mathcal{S}_L \mathcal{R}_L \mathcal{T}_L \right \} \nonumber \\
  & & \mathcal{S}_i, \mathcal{T}_i \in \mathcal{P}_{\text{free}}, \quad \mathcal{R}_i \in \overline{\mathcal{R}}.
\end{eqnarray}

For Bell scenarios, this ideal reads:
\begin{multline}
\mathcal{I}^{\text{Bell}}=\bigg \langle \mathcal{A}_{a|x}\mathcal{B}_{b|y}-\mathcal{B}_{b|y}\mathcal{A}_{a|x}, \\
\mathcal{A}_{a|x} \mathcal{A}_{a'|x} - \delta_{a,a'} \mathcal{A}_{a|x},
  \mathcal{B}_{b|y} \mathcal{B}_{b'|y} - \delta_{b,b'} \mathcal{B}_{b|y}
\bigg \rangle
\label{eq:ideal}
\end{multline}
enumerated over $a,a'=0,...,n_A-2$, $b,b'=0,...,n_B-2$, $x=0,...,n_X-1$ and $y=0,...,n_Y-1$.
From $\mathcal{I}$, we now define an equivalence relation $\sim$ on $\mathcal{P}_{\text{free}}$.
For $\mathcal{S}, \mathcal{T} \in \mathcal{P}_{\text{free}}$, we have
\begin{equation}
  \mathcal{S} \sim \mathcal{T} \quad \Leftrightarrow \quad \mathcal{S}-\mathcal{T} \in \mathcal{I}.
\end{equation}

Next, we describe a notion of working with noncommutative polynomials up to equivalence.

\subsection{Quotient polynomial ring}
\label{subsec: quotient poly ring}
For an element $\mathcal{S} \in \mathcal{P}$, we write $[\mathcal{S}] =\mathcal{S}+\mathcal{I}= \{ \mathcal{S}+ I : I \in \mathcal{I} \}$ the equivalence class of $\mathcal{S}$. 
We denote by $\mathcal{P}=\mathcal{P}_{\text{free}} /\mathcal{I}$ the set of such equivalence classes. 
One verifies $[\mathcal{S}] + [\mathcal{T}] = (\mathcal{S}+\mathcal{I}) + (\mathcal{T}+\mathcal{I}) = (\mathcal{S}+\mathcal{T}) +\mathcal{I}= [\mathcal{S}+\mathcal{T}]$ and similarly for multiplication. 
Thus $\mathcal{P}$ is also a noncommutative ring named the {\em quotient ring}. 
One verifies that the operation $\ast$ is well-defined on $\mathcal{P}$ if $\mathcal{I}$ is closed under the $\ast$ operation.

Now, working with elements of $\mathcal{P}$ as equivalence classes is pretty cumbersome.
Thus, for each element $[\mathcal{S}] \in \mathcal{P}$, we single out a canonical representative $\overline{\mathcal{S}} \in \mathcal{P}_{\text{free}}$, obtained by applying rewriting rules.
The rewriting system should be confluent, in that the same final $\overline{\mathcal{S}}$ is obtained after rewriting any element in $[\mathcal{S}]$.
For Eqs.~\eqref{eq:commutation}-\eqref{eq:projector}, the rewriting rules are, for all $a,a',b,b',x,y$:
\begin{eqnarray}
  \mathcal{B}_{b|y} \mathcal{A}_{a|x} & \rightarrow &
  \mathcal{A}_{a|x} \mathcal{B}_{b|y}, \nonumber \\
  \mathcal{A}_{a|x} \mathcal{A}_{a'|x} & \rightarrow & 
  \delta_{a,a'} \mathcal{A}_{a|x}, \nonumber \\
  \mathcal{B}_{b|y} \mathcal{B}_{b'|y} & \rightarrow & 
  \delta_{b,b'} \mathcal{B}_{b|y}.
  \nonumber
\end{eqnarray}
The confluence of this rewriting system is part of the folklore, but can be derived using the theory of noncommutative Groebner bases~\cite{Mora_1994}.

Let us take an example as an illustration. For the CHSH scenario, the free polynomial ring is given by $\mathcal{P}^{\text{CHSH}}_{\text{free}} :=\mathbb{C}[ \{ \mathcal{A}_{0|0},\mathcal{A}_{0|1},\mathcal{B}_{0|0},\mathcal{B}_{0|1}\} ]$. The monomial $\mathcal{B}_{0|0}\mathcal{A}_{0|1}\mathcal{B}_{0|0}\mathcal{A}_{0|0}$ has a representant in the quotient ring $[\mathcal{A}_{0|1}\mathcal{A}_{0|0}\mathcal{B}_{0|0}]$.
The CHSH inequality takes the form:
\begin{equation}
\label{eq:chshobjective}
    \mathcal{E}^{\text{CHSH}} = 2 + 4 \sum_{x,y=0}^1 (-1)^{(x+1)y} [\mathcal{A}_{0|x} \mathcal{B}_{0|y}] - 4 [\mathcal{A}_{0|1}] - 4 [\mathcal{B}_{0|0}].
\end{equation}

For the rest of this manuscript, we will work with canonical representatives of the quotient ring $\mathcal{P}$, and drop the $[\ldots]$ bracket representing equivalence classes.

Our optimization problem is thus, with $\mathcal{E} \in \mathcal{P}$ and $E = \left \langle \psi \middle | \mathcal{E} \middle | \psi \right \rangle$:
\begin{equation}
\label{eq:maxE}
    \sup_{\left | \psi \right \rangle, \{\mathcal{X}_i\}_i}  E \text{ s.t. } \left \langle \psi \middle | \psi \right \rangle = 1.
\end{equation}

The objective to be maximized needs to be real: to ensure this, we require $\mathcal{E}^\ast = \mathcal{E}$.

\section{Semidefinite programming relaxations}
\label{sec:semidefiniterelaxations}
Computing the quantum supremum $\mu_{\mathcal{Q}}$ is difficult in general.
However, any explicit solution of the optimization
problem~{\eqref{eq:optimizationproblem}} with $\lambda = E^{\text{Bell}}$
provides a lower bound $\lambda \leqslant \mu_Q$. Such a lower bound can be found using iterative optimization techniques~\cite{Pal_2009}, or educated analytical guesses~\cite{Collins_2002,Acin_2002}.

Upper bounds $\mu_{\text{Q}} \leqslant u$ are derived usually using
semidefinite relaxations. Such relaxations can be understood using two
equivalent viewpoints: {\em moment-based} approaches and
{\em sum-of-squares} certificates.

\subsection{Moment-based relaxations}
\label{sec:momentbased}
Let $\vec{\mathcal{Q}} \in \mathcal{P}^n$, with coefficients $\mathcal{Q}_i$, be a vector of linearly independent elements of $\mathcal{P}$ for $i = 1, \ldots, n$, which we name the {\em generating sequence}. 
We define the matrix $\Xi \in \mathcal{P}^{n \times n}$ with coefficients $\Xi_{i j} = \mathcal{Q}_i \mathcal{Q}_j^{\ast}$.
We have
\begin{equation}
\label{eq:defXi}
    \Xi = \vec{\mathcal{Q}} \cdot \vec{\mathcal{Q}}^\dag\;,
\end{equation}
where $(\cdot)^\dagger := \big ((\cdot)^\ast \big)^\top$ is the conjugate transpose, $(\cdot)^\top$ the matrix transpose, and where the complex conjugate $(\cdot)^\ast$ on $\mathcal{P}$ is the involution.
We also write $(\cdot)^{-\dagger} = \big ( (\cdot)^{-1} \big )^\dag$.
Altogether, the coefficients $\Xi_{i j}$ span a vector subspace of $\mathcal{P}$.

We consider a basis of that subspace consisting of the linearly independent elements $(\mathcal{M}_k)_k$ for $k = 0, \ldots, m$.
For later convenience, we impose three conditions on this sequence.
First, the sequence starts with
\begin{equation}
\label{eq:condfirstone}
    \mathcal{M}_0 = 1\;,
\end{equation}
and we number the $\mathcal{M}_k$ starting from $0$ to treat scalars in $\mathcal{P}$ in a differentiated manner.
Second, the elements of the sequence are self-adjoint:
\begin{equation}
\label{eq:condselfadjoint}
    \mathcal{M}_k = \mathcal{M}_k^\dagger, \quad \forall k.
\end{equation}
Indeed, if, $\mathcal{M}_k$ is not self-adjoint, then we define $\mathcal{M}'_k = \mathcal{M}_k + \mathcal{M}_k^\dagger$ and $\mathcal{M}''_k = \sqrt{-1} (\mathcal{M}_k - \mathcal{M}_k^\dagger)$, such that $2 \mathcal{M}_k =\mathcal{M}_k' - \sqrt{-1} \mathcal{M}_k''$.
Third, the generating sequence $\vec{\mathcal{Q}}$ contains enough relevant elements such that the objective $\mathcal{E}$ can be expanded over the resulting $\mathcal{M}_k$ with $b_0 \in \mathbb{R}$ and $\vec{b} \in \mathbb{R}^m$ (real coefficients because $\mathcal{E} = \mathcal{E}^\dagger$):
\begin{equation}
  \label{eq:expandobjective} \mathcal{E}= b_0 + \sum_{k=1}^m b_k \mathcal{M}_k \; .
\end{equation}

The linear independence of the $\mathcal{M}_k$ defines uniquely the matrices $A_k \in \mathbb{C}^{n \times n}$ such that:
\begin{equation}
  \label{eq:expandmatrix} \Xi = \sum_{k=0}^m A_k \mathcal{M}_k \;,
\end{equation}
and from $\Xi^\dagger = \Xi$ we get $A_k^\dagger = A_k$ for all $k$.

The name of the technique comes from considering the {\em moments}
\begin{equation}
\label{eq:defmoments}
    y_k = \langle \psi | \mathcal{M}_k | \psi \rangle.
\end{equation}
We have by normalization of $|\psi \rangle$ that $y_0 = \langle \psi | 1 | \psi \rangle = 1$, and we write $\vec{y} \in \mathbb{R}^m$ for the remaining non-constant moments.
The matrix $Z = Z^{\dag} \in \mathbb{C}^{n \times n}$, with $Z_{i j} = \langle \psi | Q_i Q_j^{\ast} | \psi \rangle$ can be written
\begin{equation}
    \label{eq: Z}
    Z = \left \langle \psi \middle | \Xi \middle | \psi \right \rangle = A_0 + \sum_{k = 1}^m y_k A_k\;.
\end{equation}

\begin{proposition}
\label{prop: Z sdp}
  The matrix $Z$ is semidefinite positive (written $Z \succeq 0$).
\end{proposition}

\begin{proof}
  Let $\vec{w} \in \mathbb{C}^n$ be arbitrary. Then $\vec{w}^{\dag} Z \vec{w}
  = \langle \psi | \vec{w}^\dag \vec{\mathcal{Q}} \vec{\mathcal{Q}}^\dag \vec{w} | \psi \rangle = \langle \psi | \mathcal{W} \mathcal{W}^{\ast} | \psi \rangle
  \geqslant 0$ for $\mathcal{W}= \vec{w}^\dag \cdot \vec{\mathcal{Q}}$, and thus $Z$ is semidefinite positive.
\end{proof}

We combine the maximization~\eqref{eq:maxE} and the expansion~\eqref{eq:expandobjective}:
\begin{equation}
\label{eq:momentobj}
    E = b_0 + \sum_{k=1}^m b_k y_k = b_0 + \vec{b}^\top \vec{y},
\end{equation}
to finally write our semidefinite positive relaxation:
\begin{equation}
  \label{eq:momentsdp}
  \begin{split}
      d^\star = \max_{\vec{y} \in \mathbb{R}^m} \;\; &b_0 + \vec{b}^\top \vec{y} \\
      \text{ s.t. } \;\; & A_0 + \sum_{k = 1}^m y_k A_k \succeq 0.
  \end{split}
\end{equation}
It is a relaxation as the constraint $Z \succeq 0$, for a given generating
sequence, does not necessarily encode all the constraints that the moments satisfy.

\subsection{Sum-of-squares relaxations}
\label{sec:sumofsquares}

Another approach is to minimize $\mu \in \mathbb{R}$ such that $\mu - E =\langle \psi | \mu -\mathcal{E} | \psi \rangle \geqslant 0$. 
Writing $\mathcal{F}= \mu - \mathcal{E}$, we know that $\langle \psi | \mathcal{F} |
\psi \rangle \geqslant 0$ if $\mathcal{F}$ is a sum-of-squares (SOS):
\begin{equation}
    \label{eq:formalsumofsquares}
  \mathcal{F}= \sum_{\ell = 1}^r \mathcal{T}_{\ell}^{\ast} \mathcal{T}_{\ell}, \qquad \mathcal{T}_{\ell} \in \mathcal{P}.
\end{equation}
Any such $\mathcal{F}$ is a certificate that $\mu$ is an upper bound on the
objective value $E$ of the original problem. 
The goal is to find the tighest bound $\mu$, i.e. to minimize $\mu$ such that $\mathcal{F}= \mu - \mathcal{E}$ is a SOS.

To tackle this problem numerically, we restrict the elements $\mathcal{T}_{\ell}$ to be a linear combination of a given generating set $(\mathcal{Q}_i)_i$ for $i = 1, \ldots, n$.
For a given generating set, the set of potential SOS certificates is smaller than for the unrestricted case, and thus the bound $\mu$ we obtain may be not as tight as it could be. 
It is nevertheless valid.

We write $\mathcal{T}_\ell = \vec{\chi}_\ell^\dagger \cdot \vec{\mathcal{Q}}$, where $\vec{\mathcal{Q}}$ is the column vector with coefficients $(\mathcal{Q}_i)_i$.
Our problem becomes
\begin{eqnarray}
    \label{eq:Xsdp}
    \mathcal{F} &=& \sum_{\ell=1}^r (\vec{\mathcal{Q}}^\dag \cdot \vec{\chi}_\ell ) (\vec{\chi}_\ell^\dag \cdot \vec{\mathcal{Q}}) \nonumber \\
    &=&\vec{\mathcal{Q}}^\dag \cdot \underbrace{\left ( \sum_{\ell=1}^r  \vec{\chi}_\ell \cdot \vec{\chi}_\ell^\dag \right )}_{:= X} \cdot \vec{\mathcal{Q}} \nonumber \\
    & = & \operatorname{tr}[\Xi X]
\end{eqnarray}
and by definition $X$ is a semidefinite positive matrix.

We use again a sequence $(\mathcal{M}_k)_k$, $k = 0, \ldots, m$ satisfying $\mathcal{M}_0 = 1$, such that $\mathcal{E}$ can be expanded using $b_k$ as in Eq.~{\eqref{eq:expandobjective}}
\begin{equation}
 \mathcal{F} = \mu - b_0 - \sum_{k=1}^m b_k \mathcal{M}_k,
 \end{equation}
while $\Xi$ can be expanded in $A_k$ using Eq.~{\eqref{eq:expandmatrix}}:
 \begin{equation}
 \mathcal{F} = \operatorname{tr}[\Xi X] = \operatorname{tr}[A_0 X] + \sum_{k=1}^m \operatorname{tr}[A_k X] \mathcal{M}_k\;.
\end{equation}

We match the terms of two last equations, as the $\mathcal{M}_k$ are linearly independent.
We have $\mu = b_0 + \operatorname{tr}[A_0 X]$, which we minimize, and the equality constraints $\operatorname{tr} [A_k X] = - b_k$, with $X \succeq 0$. 
Our optimization problem is thus:
\begin{align}
  \label{eq:sossdp} 
   p^\star = \min_{X = X^{\dag} \in \mathbb{C}^{n \times n}}\;\; & b_0 + \operatorname{tr} [A_0 X]&  \nonumber\\
    \text{s.t. } \;\; & \operatorname{tr} [A_k X] = - b_k& k=1,\ldots,m, \nonumber\\
    &X \succeq 0,& 
\end{align}
which is also a semidefinite program.

\subsection{Convergence and duality}

The semidefinite programs~\eqref{eq:momentsdp} and~\eqref{eq:sossdp} correspond to the standard primal/dual forms of a semidefinite program~\cite{Vandenberghe_1996} with two small deviations: an additional constant factor $b_0$ and sign conventions.
Most semidefinite programming solvers work with a primal and dual pair of variables $(X, \vec{y})$, so it is not necessary to choose which formulation to solve.
When the value of the primal and dual SDPs coincide, i.e. $p^\star=d^\star$, then we say that {\em strong duality} holds. 
As the {\em generating set} grows in size, the SDP solution $u = p^\star = d^\star$ converges usually to the quantum supremum $\mu_{\mathcal{Q}}$ (technicalities are discussed in~\cite{Navascues_2008,Pironio_2010}).

Let us show what strong duality implies on the optimization variables of Eq.~\eqref{eq:sossdp} and Eq.~\eqref{eq:momentsdp}. 
We write the Lagrangian corresponding to our primal/dual problem:
\begin{equation}
        L = b_0 + \sum_{k=1}^m b_k y_k + \operatorname{tr}\big [ X(A_0+\sum_{k=1}^m y_k A_k ) \big]
\end{equation}
Next, we consider the supremum of $L$ over the variables of~\eqref{eq:momentsdp}:
\begin{equation}
    S = \sup_{\vec{y}} L.
\end{equation}
The last term of $L$ is for any particular solution $\vec{y}$ of the SDP~\eqref{eq:momentsdp} positive if $X\succeq0$. 
The two first terms in $L$ correspond to the objective function of the dual SDP. 
It follows that $S$ is an upper bound on the value of the dual. 
Let us assume that strong duality holds, then $S = L$ and the last positive term must vanish. 
Hence, if strong duality holds, we have complementarity:
\begin{equation}
    X(A_0+\sum_{k=1}^m y_k A_k) = X Z = \mathbf{0}.
    \label{eq: complementarity}
\end{equation}

\subsection{Certificates}
A solution $X$ of the SOS relaxation is a sort of certificate, as it provides an upper bound $p^\star$ on the maximum quantum value $\mu_{\mathcal{Q}}$.
The verification of such a certificate needs the following.
\begin{itemize}
    \item The matrix $X$ is semidefinite positive. The computation of the eigenvalues of $X$ can be expensive in exact arithmetic; for small $X$, the principal minor test~\cite{Prussing1986} is faster.
    \item The matrix $X$ encodes a certificate for the inequality $\mathcal{E}$ with the proper bound, i.e. that $\operatorname{tr}[\Xi X] = p^\star - \mathcal{E}$.
    This requires the computation of the matrix $\Xi$.
\end{itemize}

Alternatively, one looks for an explicit sum-of-squares decomposition of the form~\eqref{eq:formalsumofsquares}, which has the big advantage of being self-contained, and its validity is easily checked by expanding the products.
To obtain it, we factor $X \succeq 0$ using the Cholesky decomposition
\begin{equation}
\label{eq:chol}
  X = LDL^{\dag} = L \sqrt{D} \sqrt{D} L^\dag,
\end{equation}
where $L \in \mathbb{C}^{n \times r}$ and $D \in \mathbb{R}^{r \times r}$ is a
diagonal matrix with nonnegative elements.
This formulation postpones the extraction of square roots, requiring only field operations for its computation.
This is preferable when doing exact symbolical computations.
We put the decomposition in~\eqref{eq:Xsdp}:
\begin{equation}
    \mathcal{F}  = \vec{\mathcal{Q}}^\dag L \sqrt{D} \underbrace{\sqrt{D} L^\dag \vec{\mathcal{Q}}}_{=\vec{\mathcal{T}}} = \vec{\mathcal{T}}^\dag \vec{\mathcal{T}} 
\end{equation}
which corresponds to the form~\eqref{eq:formalsumofsquares} if we define 
\begin{equation}
\label{eq: certificateElements}
\mathcal{T}_\ell = \sqrt{D_{\ell\ell}} \sum_i L^\ast_{i \ell} \mathcal{Q}_i .   
\end{equation}
Note that in~\eqref{eq:formalsumofsquares}, one can take the factor $D_{\ell \ell}$ out of the $\mathcal{T}^\ast_\ell \mathcal{T}_\ell$ product, and thus we never need to extract the square root.

In the present paper, we are interested in Bell inequalities where the hierarchy converges after a finite number of steps and where the coefficients of $X$ are somehow easy to handle (for example, in a suitable extension of the rationals).
We will present an explicit SOS decomposition for most of the examples we consider, except in the case where the decomposition is too cumbersome to be checked by hand.
In such cases, we provide the solution $X$ and check its validity using the principal minor test.

\subsection{Example: CHSH}
\label{sec:chshnonsymexample}
We remember the inequality~\eqref{eq:chshobjective}:
\begin{eqnarray}
\label{eq:chshexpanded}
    \mathcal{E} & = & 2 - 4 \mathcal{A}_{0|1} - 4 \mathcal{B}_{0|0} + 4\mathcal{A}_{0|0}\mathcal{B}_{0|0} + 4\mathcal{A}_{0|1}\mathcal{B}_{0|0} \nonumber \\
    & & - 4 \mathcal{A}_{0|0} \mathcal{B}_{0|1} + 4 \mathcal{A}_{0|1} \mathcal{B}_{0|1}
\end{eqnarray}
and consider a semidefinite relaxation using the generating sequence
\begin{equation}
\label{eq: CHSHQ}
    \vec{\mathcal{Q}} = \left( 1,\mathcal{A}_{0|1},\mathcal{A}_{1|1},\mathcal{B}_{0|1},\mathcal{B}_{1|1} \right)^\top \;,
\end{equation}
which spans the space of all polynomials of degree at most one.
The symbolic matrix $\Xi_{ij} = \mathcal{Q}_i \mathcal{Q}_j^\ast$ takes the form
 \begin{equation}
     \scalemath{1}{
 \Xi = 
     \begin{pmatrix}
        1 & \mathcal{A}_{0|0} & \mathcal{A}_{0|1} & \mathcal{B}_{0|0} & \mathcal{B}_{0|1} \\
        \mathcal{A}_{0|0} & \mathcal{A}_{0|0} & \mathcal{A}_{0|0}\mathcal{A}_{0|1} & \mathcal{A}_{0|0}\mathcal{B}_{0|0} & \mathcal{A}_{0|0}\mathcal{B}_{0|1} \\
        \mathcal{A}_{0|1} & \mathcal{A}_{0|1}\mathcal{A}_{0|0}& \mathcal{A}_{0|1} & \mathcal{A}_{0|1}\mathcal{B}_{0|0} & \mathcal{A}_{0|1}\mathcal{B}_{0|1} \\
        \mathcal{B}_{0|0} & \mathcal{A}_{0|0}\mathcal{B}_{0|0} & \mathcal{A}_{0|1}\mathcal{B}_{0|0} & \mathcal{B}_{0|0} & \mathcal{B}_{0|0}\mathcal{B}_{0|1} \\
        \mathcal{B}_{0|1} & \mathcal{A}_{0|0}\mathcal{B}_{0|1} & \mathcal{A}_{0|1}\mathcal{B}_{0|1} & \mathcal{B}_{0|1}\mathcal{B}_{0|0} & \mathcal{B}_{0|1}
    \end{pmatrix}.}
\end{equation}
Not all entries of $\Xi$ are linearly independent. Let $\mathcal{M}=$ $\{1$, $\mathcal{A}_{0|0}$, $\mathcal{A}_{0|1}$, $\mathcal{B}_{0|0}$, $\mathcal{B}_{0|1}$, $\mathcal{A}_{0|0}\mathcal{A}_{0|1}+\mathcal{A}_{0|1}\mathcal{A}_{0|0}$, $i(\mathcal{A}_{0|0}\mathcal{A}_{0|1}-\mathcal{A}_{0|1}\mathcal{A}_{0|0})$, $\mathcal{B}_{0|0}\mathcal{B}_{0|1}+\mathcal{B}_{0|1}\mathcal{B}_{0|0}$, $i(\mathcal{B}_{0|0}\mathcal{B}_{0|1}-\mathcal{B}_{0|1}\mathcal{B}_{0|0})$, $\mathcal{A}_{0|0}\mathcal{B}_{0|0}$, $\mathcal{A}_{0|1}\mathcal{B}_{0|0}$, $\mathcal{A}_{0|0}\mathcal{B}_{0|1}$, $\mathcal{A}_{0|1}\mathcal{B}_{0|1}\}$ be the set of linearly independent elements of $\Xi$. The set $\mathcal{M}$ has of a cardinality of 13 (with index $k=0,...,12$).
The inequality~\eqref{eq:chshexpanded} is rewritten as in Eq.~\eqref{eq:expandobjective}:
\begin{equation}
    \mathcal{E} = 2 \mathcal{M}_0 - 4 \mathcal{M}_2 - 4 \mathcal{M}_3 + 4 \mathcal{M}_9 - 4 \mathcal{M}_{10} + 4 \mathcal{M}_{11} + 4 \mathcal{M}_{12}\;.
\end{equation}

We can expand $\Xi$ with respect to this set such as in Eq.~\eqref{eq:expandmatrix}.
Let us consider the the moments $y_k = \bra{\psi}\mathcal{M}_{k}\ket{\psi}$ with $k\in\{0,1,...,12\}$. As $\mathcal{M}_0 = 1$, we have $y_0=1$.
The moment matrix in Eq.~\eqref{eq: Z} reads:
\begin{equation}
    Z = 
        \begin{pmatrix}
            1 & y_1 & y_2 & y_3 & y_4 \\
            y_1 & y_1 & \frac{y_5-i y_6}{2} & y_9 & y_{10} \\
            y_2 & \frac{y_5+i y_6}{2} & y_2 & y_{11} & y_{12} \\
            y_3 & y_9 & y_{11} & y_3 & \frac{y_7-iy_8}{2} \\
            y_4 & y_{10} & y_{12} & \frac{y_7+iy_8}{2}  & y_4
        \end{pmatrix}.
\end{equation}
Using Eq.~\eqref{eq: Z} it is straightforward to recover the matrices $(A_k)_k$. 
As already proven in Prop.~\ref{prop: Z sdp}, $Z$ is semidefinite positive. 
The optimization problem takes the form :
\begin{equation}
    \begin{split}
        d^\star=\max_{\vec{y}} \;\; &2-4y_2-4y_3+4y_9-4y_7+4y_{11}+4y_{12}\\
        \text{s.t.} \;\; &Z \succeq 0.
    \end{split}
\end{equation}
Solving the SDP numerically yields the well known Tsirelson bound $d^\star=2\sqrt{2}$~\cite{cirel1980}. 
The optimal matrix $Z^\star$ is of rank 3 and hence due to the complementarity slackness~\eqref{eq: complementarity}, we expect $X^\star$ to be of rank 2.

From the numerical solution of the primal form~\eqref{eq:sossdp}, we guess using the ring of integers $\mathbb{Z}\{\sqrt{2}\}$ the optimum $p^\star=2\sqrt{2}$, and the form of the SDP variable:
\begin{equation}
\label{eq:primalsolutionchsh}
    X^\star = 
        \begin{pmatrix}
             2 \left(\sqrt{2}-1\right) & -\sqrt{2} & 2-\sqrt{2} & 2-\sqrt{2} & -\sqrt{2} \\
             -\sqrt{2} & 2 \sqrt{2} & 0 & -2 & 2 \\
             2-\sqrt{2} & 0 & 2 \sqrt{2} & -2 & -2 \\
             2-\sqrt{2} & -2 & -2 & 2 \sqrt{2} & 0 \\
            -\sqrt{2} & 2 & -2 & 0 & 2 \sqrt{2} \\
         \end{pmatrix}
.
\end{equation}

Let us compute the Choleski decomposition of $X^\star$ 
\begin{equation}
\label{eq:LDLchsh}
        L = \begin{pmatrix}
             1 & 0 \\
             -\frac{c_+}{\sqrt{2}} & 1 \\
             \frac{1}{\sqrt{2}} & c_+  \\
             \frac{1}{\sqrt{2}} & -c_+ \\
             -\frac{c_+}{\sqrt{2}} & -1 
            \end{pmatrix}, \;\;
        D = \begin{pmatrix}
                 2 c_- & 0\\
                 0 & c_- 
             \end{pmatrix}
\end{equation}
where $c_\pm=\sqrt{2} \pm 1$.
From the Choleski decomposition, one sees that the SOS will only contain two non-zero terms. 
It is as previously already mentioned, what was expected due to the complementarity slackness with respect to the moment-based approach. 
Using Eq.~\eqref{eq: certificateElements}, it is straightforward to recover the certificate elements and the SOS reads
\begin{eqnarray}
    \label{eq: CHSHSOSmomentbasedrelaxation}
        2\sqrt{2}-\mathcal{E} & = & 
        c_- \bigg[ \|\sqrt{2}-c_+
   \mathcal{A}_{0|0}+\mathcal{A}_{0|1}+\mathcal{B}_{0|0}-c_+ \mathcal{B}_{0|1} \|^2 \nonumber \\
   & & + \|\mathcal{A}_{0|0}+c_+ \mathcal{A}_{0|1}-c_+ \mathcal{B}_{0|0}-\mathcal{B}_{0|1}\|)^2\bigg].
\end{eqnarray}

In general, it is not easy to guess the analytical form of the SDP variables from the numerical results.
In the next section, we will harness the symmetries present in the optimization problem to reduce the number of variables and make the recovery of an analytical solution easier.

\section{Symmetries: definitions}
\label{sec:symmetriesdefs}
Many Bell inequalities have symmetries, as do many polynomial optimization problems.
Those symmetries reduce drastically the computational burden of solving the corresponding semidefinite relaxations.
We formalize below an approach that can be used with general noncommutative polynomial rings.
More effective approaches are available in particular cases, such as commutative rings~\cite{Gatermann_2004}.

\subsection{Real versus complex solutions}
\label{sec:realvscomplex}
We start with a warm-up.
Some semidefinite problems have symmetry under complex conjugation.
For example, when optimizing the violation of a Bell inequality, one can always assume that the operators and the state are real.
We generalize this case.

Assume that the following conditions hold.
First, $A_0 = A_0^\ast$ (remember that $(\cdot)^\ast$ is the complex conjugate). Second, there exists a matrix $F \in  \mathbb{R}^{m \times m}$ such that
\begin{eqnarray}
    Z^\ast & = & A_0 + \sum_{k=1}^m y_k A_k^\ast \nonumber \\
    & = & A_0 + \sum_{k=1}^m (F \vec{y})_k A_k\;.
\end{eqnarray}
Then, $F$ represents the complex conjugation expressed using the moments $\vec{y}$.
Third, we have $\vec{b}^\top = \vec{b}^\top F$: the objective is invariant under complex conjugation.

If those conditions hold, for any feasible (respectively optimal) solution, $\vec{y}$, we have that $F \vec{y}$ is also a feasible (respectively optimal) solution. By convexity, we derive a new feasible (respectively optimal) solution $\overline{y}$:
\begin{equation}
\label{eq:averagingconjugate}
    \overline{y} = \frac{\vec{y} + F \vec{y}}{2}\;.
\end{equation}
In the example of Section.~\ref{sec:chshnonsymexample}, as with all Bell inequality optimization problems, we have this symmetry. The matrix $F$ reads $F = \operatorname{diag}(1,1,1,1,1,1,-1,1,-1,1,1,1,1)$, and averaging the solution and its conjugate as in~\eqref{eq:averagingconjugate} simply sets $y_6 = y_8 = 0$, thus reducing the problem size.

We now consider more complicated symmetries at the level of the polynomial ring, and we will use again the idea of averaging over a set of solutions as in~\eqref{eq:averagingconjugate}.

\subsection{Symmetries of polynomial rings}
\label{subsec:sym poly ring}
Let $\mathcal{P}$ be a polynomial ring as defined in Section~\ref{subsec: quotient poly ring}.
An endomorphism of $\mathcal{P}$ is a linear map $\varphi : \mathcal{P}\rightarrow \mathcal{P}$ which additionally preserves the multiplicative
\begin{equation}
  \label{eq:morphismmultiplicative} \varphi (\mathcal{S}\mathcal{T}) = \varphi
  (\mathcal{S}) \varphi (\mathcal{T})
\end{equation}
and adjoint structure $\varphi (\mathcal{S}^{\dag}) = \varphi (\mathcal{S})^{\dag}$
for all $\mathcal{S}, \mathcal{T} \in \mathcal{P}$. 
From Eq.~{\eqref{eq:morphismmultiplicative}}, we deduce that the morphism $\varphi$ is fully defined by the operator images $\varphi (\mathcal{X}_1),\ldots, \varphi (\mathcal{X}_N)$. We now consider {\em affine} endomorphisms which satisfy $\deg \varphi (\mathcal{X}_i) \leqslant 1$ and $\varphi(1) = 1$. 
As symmetries are reversible transformations, we require $\phi$ to be invertible, i.e. that it is an affine {\em automorphism} of $\mathcal{P}$.
Any affine automorphism can be described using the inverse $\omega^{-1}$ of a $(N + 1) \times (N + 1)$ matrix $\omega$:
\begin{equation}
\label{eq:affinegroupmatrix}
\begin{pmatrix}
     1\\
     \varphi (\mathcal{X}_1)\\
     \vdots\\
     \varphi (\mathcal{X}_N)
\end{pmatrix} =
\begin{pmatrix}
     1 & 0 & \!\cdots\! & 0\\
     (\omega^{-1})_{21} & (\omega^{-1})_{22} &  & (\omega^{-1})_{2 N}\\
     \vdots &  &  & \vdots\\
     (\omega^{-1})_{N 1} & (\omega^{-1})_{N 2} & \!\cdots\! & (\omega^{-1})_{N N}
\end{pmatrix}
\begin{pmatrix}
     1\\
     \mathcal{X}_1\\
     \vdots\\
     \mathcal{X}_N
\end{pmatrix}
\end{equation}
where the first row of both $\omega$ and $\omega^{-1}$ is fixed to $(1, 0, \ldots, 0)$.
Such automorphisms, equipped with composition of morphism form $\Omega$, the (affine) automorphism group of $\mathcal{P}$.

Why the use of the inverse matrix in the definition~\eqref{eq:affinegroupmatrix}?
We write $\vec{\mathcal{X}} = (1, \mathcal{X}_1, \ldots, \mathcal{X}_N)^\top$, and, by convention, $\varphi(\vec{\mathcal{X}}) = (1, \varphi(\mathcal{X}_1), \ldots, \varphi(\mathcal{X}_N))^\top$ is applied element-by-element.
We consider two affine automorphisms $\varphi_1$ and $\varphi_2$ with corresponding matrices $\omega_1^{-1}$ and $\omega_2^{-1}$ such that
\begin{equation}
    \varphi_1(\vec{\mathcal{X}}) = \omega_1^{-1} \mathcal{X}, \qquad
    \varphi_2(\vec{\mathcal{X}}) = \omega_2^{-1} \mathcal{X}.
\end{equation}
Hence, the composition of morphisms is compatible with matrix multiplication:
\begin{align}
  \label{eq:matrixautomorphism}
    (\varphi_1 \circ \varphi_2) (\vec{\mathcal{X}}) & = \varphi_1 ( \omega_2^{-1} \vec{\mathcal{X}} ) = \omega_2^{-1} \varphi_1(\vec{\mathcal{X}}) = \omega_2^{-1} \omega_1^{-1} \vec{\mathcal{X}} \nonumber \\
    & = (\omega_1\omega_2)^{-1} \vec{\mathcal{X}}.
\end{align}

In the CHSH example~\eqref{eq:chshexpanded}, the automorphism group is generated by the three operations:
\begin{equation}
\label{eq:CHSHexamplesym}
\underbrace{\begin{array}{rl}
    A_{0|0} & \rightarrow 1 - A_{0|0} \\
    A_{0|1} &\rightarrow A_{0|1} \\
    B_{0|0} &\rightarrow B_{0|0} \\
    B_{0|1} & \rightarrow B_{0|1}
      \end{array}}_{\varphi_2}\;,
\underbrace{\begin{array}{rl}
    A_{0|0} & \rightarrow A_{0|1} \\
    A_{0|1} &\rightarrow A_{0|0} \\
    B_{0|0} &\rightarrow B_{0|0} \\
    B_{0|1} & \rightarrow B_{0|1}
      \end{array}}_{\varphi_3}\;,
\underbrace{\begin{array}{rl}
    A_{0|0} & \rightarrow B_{0|0} \\
    A_{0|1} &\rightarrow B_{0|1} \\
    B_{0|0} &\rightarrow A_{0|0} \\
    B_{0|1} & \rightarrow A_{0|1}
      \end{array}}_{\varphi_4}\;.
\end{equation}
Note that we use the indices $2,3,4$ so that these generators are compatible with the notation~\eqref{eq:cglmp:automorphismgenerators} used later.

Alternatively, one can describe this group as a matrix group generated by
\begin{equation}
    \begin{aligned}
        &\omega_2^{-1} = \omega_2 = \begin{pmatrix*}[r]
        1 & 0 \\
        1 & -1
    \end{pmatrix*} \oplus \mathbb{1}_3, 
            &\omega_3^{-1}& = \omega_3 = 1 \oplus \sigma_1 \oplus \mathbb{1}_2,  \\
            &\omega_4^{-1} = \omega_4 = 1 \oplus (\sigma_1 \otimes \mathbb{1}_2), 
            &\sigma_1& := \begin{pmatrix*}[r] 0 & 1 \\ 1 & 0 \end{pmatrix*}.
    \end{aligned}
\end{equation}

Those two groups are isomorphic and of order $|G| = 128$, with presentation
$\Omega^{\text{CHSH}} = \langle \omega_2, \omega_3, \omega_4 | \omega_2^2 = \omega_3^2 = \omega_4^2 = (\omega_2 \omega_3)^4 = (\omega_2 \omega_4)^4 = (\omega_3 \omega_4)^4 = (\omega_3 \omega_4 \omega_2 \omega_4)^2 = \mathbb{1} \rangle $.

For the rest of this paper, we focus on affine automorphism groups, and describe them using the matrices~\eqref{eq:affinegroupmatrix}.

\subsection{Symmetries of optimization problems}
\label{subsec:sym opt prob}
Given a polynomial ring $\mathcal{P}$ and an objective function $\mathcal{E}$ as in Eq.~\eqref{eq:objectivefunction}, we define
\begin{equation}
\label{eq:defsymgroup}
    G = \left\{ g \in \Omega \middle | \mathcal{E} = g(\mathcal{E}) \right\}
\end{equation}
which is the symmetry group of the optimization problem.

For the CHSH objective~\eqref{eq:chshexpanded}, the symmetry group is the dihedral group order $16$; it is generated by the matrices $a$ and $x$
\begin{equation}
\label{eq:CHSHgenerators}
    a = \omega_2 \omega_3 \omega_4, \qquad x = \omega_3 \omega_4 \omega_3,
\end{equation}
with the invariance easily checked in the notation of~\eqref{eq:CHSHexamplesym}:
\begin{equation}
    \varphi_4(\varphi_3(\varphi_2(\mathcal{E})))=\varphi_3(\varphi_4(\varphi_3(\mathcal{E})))=\mathcal{E}.
\end{equation}
These matrices correspond to the standard presentation
\begin{equation}
\label{eq:Gchsh}
   G = \left \langle a, x \middle | a^8 = x^2 = x a x a = \mathbb{1} \right \rangle\;.
\end{equation}
If we viewed this dihedral group as acting on the octogon, $a$ would correspond to a rotation by 45 degrees, while $x$ would be the mirror symmetry.

Note that we do not address the question of symmetry discovery in the present paper.
In the case of Bell inequalities, we refer the reader to the permutation group approach proposed in~\cite{Rosset_2014}.

\subsection{Invariant Hermitian matrices}
How do ring automorphisms affect moment matrices and sum-of-square certificates?
In the derivations below, we assume that the vector space spanned by the generating sequence $(\mathcal{Q}_i)_i$ is closed under the action of $G$: for every $g \in G$, there is a matrix $\rho_g \in \mathbb{C}^{n \times n}$ such that:
\begin{equation}
  \label{eq:defrho}
  g(\vec{\mathcal{Q}}) = \rho_g^{-1} \vec{\mathcal{Q}}\;.
\end{equation}
As in~\eqref{eq:matrixautomorphism}, this definition ensures that the representation of the group $\rho: G \to \operatorname{GL}(\mathbb{C}, n)$ is compatible with the matrix multiplication
\begin{equation}
  (g h)(\vec{\mathcal{Q}}) = g( \rho_h^{-1} \vec{\mathcal{Q}} ) = \rho_h^{-1} g(\vec{\mathcal{Q}}) = \rho_h^{-1} \rho_g^{-1} \vec{\mathcal{Q}}\;,
\end{equation}
and as $(gh)(\vec{\mathcal{Q}}) = \rho_{gh}^{-1} \vec{\mathcal{Q}}$, we have $\rho_{g h} = \rho_g \rho_h$.
This representation has a decomposition~\cite{Serre}:
\begin{equation}
    \label{eq:rhogirrep}
    \rho_g = I \underbrace{\left( \bigoplus_{i=1}^N \mathbb{1}_{m_i} \otimes \sigma_g^{(i)} \right )}_{= \sigma_g} I^{-1}\;,
\end{equation}
where $I \in \mathbb{C}^{n \times n}$ and the representations $\sigma_g^{(i)}: G \to \operatorname{GL}(\mathbb{C}, d_i)$ are irreducible and inequivalent, but not necessarily unitary~\footnote{All irreducible representations of finite groups can be written using coefficients in the cyclotomic field; but that is not necessarily the case when requiring those irreducible representations to be unitary.}.
The $N$ irreducible representations appearing in $\rho_g$, indexed by $i$, have dimension $d_i$ and multiplicity $m_i$.

We now consider Hermitian matrices that have invariance properties under the action of $G$.
\begin{proposition}
\label{prop:blockdiag}
Let $\overline{X}, \overline{Z} \in \mathbb{C}^{n \times n}$ be Hermitian matrices satisfying:
\begin{equation}
     \overline{X} = \rho_g^\dag \overline{X} \rho_g, \quad \overline{Z} = \rho_g \overline{Z} \rho_g^\dag, \quad \forall g \in G\;.
\end{equation}
Then:
\begin{eqnarray}
    \widetilde{X}^{\text{full}} =& I^\dagger \overline{X} I &= \bigoplus_{i=1}^N \widetilde{X}^{(i)} \otimes C_i, \nonumber \\
    \widetilde{Z}^{\text{full}} =& I^{-1} \overline{Z} I^{-\dagger} &= \bigoplus_{i=1}^N \widetilde{Z}^{(i)} \otimes C_i^{-1},
\end{eqnarray}
where $C^i \in \mathbb{C}^{d_i \times d_i}$ are Hermitian and positive definite constant matrices, while the Hermitian matrices $\widetilde{X}^{(i)}, \widetilde{Z}^{(i)} \in \mathbb{C}^{m_i \times m_i}$ parameterize fully the space of $\overline{X}$ and $\overline{Z}$.

If any of the $\sigma^i$ is already unitary, the corresponding $C_i$ is a positive multiple of the identity matrix.
\end{proposition}
\begin{proof}
See Appendix~\ref{app:proof:blockdiag}.
\end{proof}

What is the impact of this last proposition? 
Writing $\overline{X} \succeq 0$ for $\overline{X}$ being semidefinite positive (SDP), we will solve later semidefinite programs where invariant matrices satisyfing the conditions of $\overline{X}$ and $\overline{Z}$ are constrained to be SDP.

In Proposition~\ref{prop:blockdiag}, we realize that as $I$ is full-rank, we have $\overline{X} \succeq 0$ if and only if $\widetilde{X}^{\text{full}} \succeq 0$, and the same holds for $\overline{Z} \succeq 0$ if and only if $\widetilde{Z}^{\text{full}} \succeq 0$ (this follows from the existence of a Cholesky decomposition, with $I$ acting as a change of basis matrix).
However, by Proposition~\ref{prop:blockdiag}, the matrices $\widetilde{X}^{\text{full}}$ and $\widetilde{Z}^{\text{full}}$ are block-diagonal, and thus the eigenvalue condition is simpler to solve.

Defining the the Hermitian matrices:
\begin{equation}
    \widetilde{X} = \bigoplus_{i=1}^N \widetilde{X}^{(i)}, \qquad \widetilde{Z} = \bigoplus_{i=1}^N \widetilde{Z}^{(i)},
\end{equation}
we have the following proposition.

\begin{proposition}
\label{prop: sdpblocs}
We have $\overline{X} \succeq 0$ if and only if $\widetilde{X} \succeq 0$.
We have as well $\overline{Z} \succeq 0$ if and only if $\widetilde{Z} \succeq 0$.
\end{proposition}
\begin{proof}
We first prove the following.
Let $A$ be a Hermitian matrix and $B$ be a positive definite Hermitian matrix.
Then $(A \otimes B) \succeq 0$ if and only if $A \succeq 0$.
We have the eigenvalue decompositions $A = U D U^\dagger$ and $B = V E V^\dagger$, so that
\begin{equation}
    A \otimes B = (U \otimes V) (D \otimes E) (U \otimes V)^\dagger
\end{equation}
and as $D = \operatorname{diag}(\vec{d})$, $E = \operatorname{diag}(\vec{e})$, we have $(A \otimes B) \succeq 0$ if and only if $\vec{d} \otimes \vec{e} \ge 0$.
As $\vec{e} > 0$, this is equivalent to $\vec{d} \geq 0$.
The proposition follows by applying this separately to the blocks of $\widetilde{X}^{\text{full}}$ or $\widetilde{Z}^{\text{full}}$.
\end{proof}

\subsection{Decomposition of a representation}
\label{subsec: decrep}
In this work, we compute the change of basis matrix $I$ following the approach described in Section 2.7 of~\cite{Serre}.
We are given a representation $\rho: G \to \operatorname{GL}(\mathbb{C}, n)$, a list $\{\sigma^{(i)}\}_i$ of irreducible representations of $G$, for $i=1, \ldots, N$.
Let $d_i$ be the dimension of the the $i$-th irreducible representation.
For each $i$, we define the family of maps $P_{\alpha \beta}^{(i)}$, where $\alpha, \beta = 1,\ldots, d_i$:
\begin{equation}
    P_{\alpha \beta}^{(i)} = \frac{d_i}{|G|} \sum_{g \in G} \left(\sigma_{g^{-1}}^{(i)}\right)_{\beta \alpha} \rho_g.
\end{equation}
From~\cite{Serre}, $P_{11}^{(i)}$ is a projector; using Gaussian elimination, we pick the first columns of $P_{11}^{(i)}$ that are linearly independent, and write them $\vec{h}^{(i)}_1, \ldots, \vec{h}^{(i)}_{m_i}$, with $m_i$ the multiplicity of $\sigma^{(i)}$ in $\rho$.
We then concatenate the following column vectors into a matrix $I^{(i),j}$:
\begin{equation}
    I^{(i),j} = \left(
    \vec{h}^{(i)}_j, P_{12} \vec{h}^{(i)}_j, \ldots, P_{1 d_i} \vec{h}^{(i)}_j
    \right)
\end{equation}
which is a map that injects the subspace corresponding to the $j$-th multiplicity space of the $i$-th irreducible representation into $\mathbb{C}^n$:
\begin{equation}
    I^{(i),j} \sigma_g^{(i)} = \rho_g I^{(i),j}.
\end{equation}
We concatenate horizontally again 
\begin{equation}
    I^{(i)} = \left( I^{(i),1)}, \ldots, I^{(i),m_i} \right)
\end{equation}
so that
\begin{equation}
    I^{(i)} \left(\mathbb{1}_{m_i} \otimes \sigma_g^{(i)}\right) = \rho_g I^{(i)}.
\end{equation}
The change of basis matrix $I$ is the final concatenation
\begin{equation}
    I = \left( I^{(1)}, \ldots, I^{(N)} \right)
\end{equation}
which has an inverse and obeys~\eqref{eq:rhogirrep}.
By convention, we write:
\begin{equation}
    I^{-1} = \begin{pmatrix}
     I^{-(1)} \\
     \vdots \\
     I^{-(N)}
    \end{pmatrix}, \qquad I^{-(i)} I^{(j)} = \begin{cases} i = j & \mathbb{1}_{d_i}, \\ i \ne j & \boldsymbol{0}. \end{cases}
\end{equation}

\section{Symmetry-adapted SDP relaxations}
\label{sec:symadaptedsdp}
We studied symmetries of optimization problems in the previous section.
When we restrict ourselves to transformations that preserve the ring structure as in Section~\ref{subsec:sym poly ring}, it means that any feasible solution of the optimization problem remains feasible under symmetry.
In Section~\ref{subsec:sym opt prob}, we further restricted those transformations to the ones leaving the objective invariant: this means that optimal solutions remain optimal under symmetry.
When considering explicit solutions of those problems, this means that any optimal solution corresponds to an orbit of solutions under the action of $G$.
However, we are studying semidefinite relaxations of the optimization problem, and those relaxations are convex.
This means that we can average over orbits under $G$, as we did for complex conjugation in Section~\ref{sec:realvscomplex}.
We will do so for both the moment-based and the sum-of-squares approaches.
As this part is pretty technical, we immediately follow the theory with its application on the CHSH example already seen before.
One can find the Mathetmatica notebook reproducing the example on Gitub \cite{Ioannou2021}.

\subsection{Symmetries of moment-based relaxations}
We consider the moment-based relaxations studied in Section~\ref{sec:momentbased}.
Let $\vec{y} = (\left \langle \psi \middle | \mathcal{M}_k \middle  | \psi \right \rangle)_k$ be a (feasible/optimal) solution of the SDP ~\eqref{eq:momentsdp}.
Then, given any transformation $g \in G$,
\begin{equation}
\vec{y}\phantom{'}' = \left( \left \langle \psi \middle | g(\mathcal{M}_k) \middle  | \psi \right \rangle) \right)_k
\end{equation}
is a (feasible/optimal) solution as well. Using Eqs.~\eqref{eq:defmoments}-\eqref{eq: Z}, we rewrite the SDP constraint:
\begin{equation}
\label{eq:Zprimemoments}
    Z' = \left\langle \psi \middle | g(\Xi) \middle | \psi \right \rangle = \sum_{k=0}^m y_k' A_k \succeq 0.
\end{equation}
We also have, using Eqs.~\eqref{eq:defXi} and~\eqref{eq:defrho}:
\begin{equation}
\label{eq:gXi}
    g(\Xi) = g(\vec{Q}) g(\vec{Q}^\dag)= \rho_g^{-1} \vec{Q} \vec{Q}^\dag \rho_g^{-\dag} = \rho_g^{-1} \Xi \rho_g^{-\dag}\;.
\end{equation}
and thus $Z' = \rho_g^{-1} Z \rho_g^{-\dag}$.

As the problem is convex, we take the orbit of solutions and define the average~\footnote{Note: we work with finite groups in this paper. In the case $G$ is a non-finite compact group, one simply defines $\Sigma_G$ as the integration over the Haar measure of $G$.} $\overline{y} \in \mathbb{R}^m$:
\begin{equation}
    \overline{y}_k =  \left( \left \langle \psi \middle | \overline{\mathcal{M}}_k \middle  | \psi \right \rangle \right)_k, \quad \overline{\mathcal{M}}_k := \frac{1}{|G|} \sum_{g\in G} g(\mathcal{M}_k),
\end{equation}
and by convexity, if $\vec{y}$ is a feasible/optimal solution, then $\overline{y}$ is a feasible/optimal solution as well.
We also average the $Z'$ and obtain, equivalently:
\begin{equation}
    \overline{Z} = \sum_{k=0}^m \overline{y}_k A_k = \frac{1}{|G|} \sum_{g\in G} \rho_g^{-1} Z \rho_g^{-\dag} = \frac{1}{|G|} \sum_{g\in G} \rho_g Z \rho_g^\dag,
\end{equation}
so that $\overline{Z} =  \rho_g \overline{Z} \rho_g^\dag$.
Here, as well as below, we send $g \to g^{-1}$ to simplify expressions being averaged.

We now make use of the following observations. First, we can restrict the variable $\vec{y}$ to the space spanned by $\overline{y}$. Second, the SDP matrix $\overline{Z}$ produced using $\overline{y}$ is invariant under $\rho$.

To reduce the number of variables in the SDP~\eqref{eq:momentsdp}, we observe that the $\overline{\mathcal{M}}_k$ are usually linearly dependent.
Nevertheless, they span a vector subspace of $\mathcal{P}$.
We write $\widetilde{\mathcal{M}}_\ell$ a basis of this subspace, indexed by $\ell = 0, \ldots, \tilde{m}$, still satisfying~\eqref{eq:condfirstone}-\eqref{eq:condselfadjoint}.
There are $\vec{w} \in \mathbb{R}^{m}$, $W: \mathbb{R}^{\tilde{m}} \to \mathbb{R}^{m}$ such that
\begin{equation}
\label{eq:basismoments}
    \begin{pmatrix}
    \overline{\mathcal{M}}_1 \\
    \ldots \\
    \overline{\mathcal{M}}_m
    \end{pmatrix}
    =
    \vec{w} + W
    \begin{pmatrix}
    \widetilde{\mathcal{M}}_1 \\
    \ldots \\
    \widetilde{\mathcal{M}}_{\tilde{m}}
    \end{pmatrix}\;,
\end{equation}
and we define the new SDP variables
\begin{equation}
    \widetilde{y}_\ell = \left \langle \psi \middle | \widetilde{\mathcal{M}}_\ell \middle | \psi \right \rangle, \qquad \overline{y} = \vec{w} + W \widetilde{y},
\end{equation}
for $\widetilde{y} \in \mathbb{R}^{\tilde{m}}$.
We define the average of $\Xi$ over the group:
\begin{equation}
\label{eq:symXi}
    \overline{\Xi} = \frac{1}{|G|}\sum_{g \in G} g(\Xi) \;,
\end{equation}
and obtain, using
\begin{equation}
        \overline{\Xi} = \sum_{k=0}^m A_k \frac{1}{|G|} \sum_{g\in G} g(\mathcal{M}_k) = \sum_{k=0}^m A_k \overline{\mathcal{M}}_k.
\end{equation}
Substituting~\eqref{eq:basismoments}, we obtain
\begin{equation}
\label{eq:overlineXiZwithredmoments}
    \overline{\Xi} = A_0' + \sum_{\ell=1}^{\tilde{m}} A'_\ell \widetilde{\mathcal{M}}_\ell, \quad \overline{Z} = A_0' + \sum_{\ell=1}^{\widetilde{m}} \widetilde{y}_\ell A_\ell' 
\end{equation}
with
\begin{equation}
    A'_0 = A_0 + \sum_{k=1}^m w_k A_k, \quad
    A'_\ell = \sum_k W_{k \ell} A_k \;.
\end{equation}
Considering~\eqref{eq:momentobj}
\begin{equation}
E = b_0 + \vec{b}^\top (\vec{w} + W \widetilde{y}) = \widetilde{b}_0 + \widetilde{b}^\top \widetilde{y}
\end{equation}
with $\widetilde{b}_0 = b_0 + \vec{b}^\top \vec{w}$ and $\widetilde{b} = \vec{b}^\top W$, we get
\begin{equation}
  \label{eq:momentsdpsympartial}
      d^\star = \max_{\widetilde{y} \in \mathbb{R}^{\widetilde{m}}} \;\; \widetilde{b}_0 + \widetilde{b}^\top \widetilde{y} \quad \text{s.t.} \quad \overline{Z} = A_0' + \sum_{\ell = 1}^{\widetilde{m}} \widetilde{y}_k A_k' \succeq 0.
\end{equation}

This semidefinite program can be simplified further.
The invariance from~\eqref{eq:gXi}
\begin{equation}
\label{eq:symXirho}
    \overline{\Xi} = \frac{1}{|G|}\sum_{g \in G} \rho_g \Xi \rho_g^\dag \quad \implies \quad \overline{\Xi} = \rho_g \overline{\Xi} \rho_g^\dag
\end{equation}
combined with the linear independence of the $\widetilde{\mathcal{M}}_\ell$ implies $A'_\ell = \rho_g A'_\ell \rho_g^\dag$ for all $\ell$ and $g \in G$.
By Proposition~\ref{prop:blockdiag}, we have
\begin{equation}
    I^{-1} A'_\ell I^{-\dag} = \bigoplus_{i=1}^N \widetilde{A}_\ell^{(i)} \otimes C_i^{-1}\;.
\end{equation}

Thus, we can replace the condition
\begin{equation}
\label{eq:Zoverline}
    \overline{Z} = \left \langle \psi \middle | \overline{\Xi} \middle | \psi \right \rangle \succeq 0
\end{equation} by $\widetilde{Z}^{\text{full}} = I^{-1} \overline{Z} I^{-\dag} \succeq 0$. By Proposition~\ref{prop: sdpblocs}, this is equivalent to
\begin{equation}
 \widetilde{Z}^{(i)} = \widetilde{A}_0^{(i)} + \sum_{\ell=1}^{\tilde{m}} \widetilde{y}_\ell \widetilde{A}_\ell^{(i)}  \succeq 0, \qquad \forall i.
\end{equation}

The final reduced semidefinite optimization problem is:
\begin{equation}
  \label{eq:momentsdpinvariant}
  \begin{split}
  d^\star  = \max_{\widetilde{y} \in \mathbb{R}^{\tilde{m}}}\;\; &\tilde{b}_0 + \sum_{\ell = 1}^{\tilde{m}} \tilde{b}_\ell \tilde{y}_\ell  \\
  \text{ s.t. } \;\; &\widetilde{Z}  =   \widetilde{A}_0 + \sum_{\ell = 1}^{\tilde{m}} \widetilde{A}_\ell \tilde{y}_\ell \succeq 0
  \end{split}
\end{equation}
where $\widetilde{A}_\ell = \bigoplus_i \widetilde{A}_\ell^{(i)}$ is a block-diagonal matrix.

\subsection{Example: CHSH}
Now, we want to illustrate the fact that symmetrizing the optimization problem enables to consider less optimization variables.

\subsubsection*{Finding a basis of invariant monomials}
We recall the symmetry group of CHSH generated by~\eqref{eq:CHSHgenerators}, which is a matrix group acting by the convention~\eqref{eq:affinegroupmatrix} on polynomials in the ring $\mathcal{P}$.
Its generators are explicitly
\begin{equation}
        a = \begin{pmatrix*}[r]
         1 & 0 & 0 & 0 &  0 \\
         1 & 0 & 0 & 0 & -1 \\
         0 & 0 & 0 & 1 &  0 \\
         0 & 1 & 0 & 0 &  0 \\
         0 & 0 & 1 & 0 &  0
         \end{pmatrix*},
         x = \begin{pmatrix*}[r]
           1 & 0 & 0 & 0 & 0 \\
           0 & 0 & 0 & 0 & 1 \\
           0 & 0 & 0 & 1 & 0 \\
           0 & 0 & 1 & 0 & 0 \\
           0 & 1 & 0 & 0 & 0
           \end{pmatrix*}.
\end{equation}
The list of elements of $G$ is
\begin{equation}
    G = \{ \mathbb{1}, a, a^2, \ldots, a^7, x, a x, a^2 x, \ldots, a^7 x \},
\end{equation}
and thus the average of an element $\mathcal{S} \in \mathcal{P}$ under $G$ is:
\begin{equation}
    \overline{\mathcal{S}} = \frac{1}{16} \sum_{j=0}^7 a (\mathcal{S} + x \mathcal{S}).
\end{equation}
The elements of the basis $\mathcal{M}$ average as 
\begin{equation}
    \begin{split}
        &\overline{\mathcal{A}_{0|0}} = \overline{\mathcal{A}_{0|1}} = \overline{\mathcal{B}_{0|0}} = \overline{\mathcal{B}_{0|1}} = \frac{1}{2},\\
        &\overline{\mathcal{A}_{0|0}\mathcal{A}_{0|1}+\mathcal{A}_{0|1}\mathcal{A}_{0|0}}  = \frac{1}{4},\\ &\overline{i(\mathcal{A}_{0|0}\mathcal{A}_{0|1}-\mathcal{A}_{0|1}\mathcal{A}_{0|0})} = 0,\\
        &\overline{\mathcal{B}_{0|0}\mathcal{B}_{0|1}+\mathcal{B}_{0|1}\mathcal{B}_{0|0}}  = \frac{1}{4}, \\ &\overline{i(\mathcal{B}_{0|0}\mathcal{B}_{0|1}-\mathcal{B}_{0|1}\mathcal{B}_{0|0})} = 0,\\
        &\overline{\mathcal{A}_{0|0}\mathcal{B}_{0|0}} = \overline{\mathcal{A}_{0|1}\mathcal{B}_{0|0}} =\overline{\mathcal{A}_{0|1}\mathcal{B}_{0|1}}= \frac{3}{8}  - \widetilde{\mathcal{M}}_1,\\
        &\overline{\mathcal{A}_{0|0}\mathcal{B}_{0|1}} = \frac{1}{8}  + \widetilde{\mathcal{M}}_1,\\
    \end{split}
\end{equation}
with $\widetilde{\mathcal{M}}_1 = \mathcal{A}_{0|1}+\mathcal{B}_{0|0}-\mathcal{A}_{0|0}\mathcal{B}_{0|0}+\mathcal{A}_{0|0}\mathcal{B}_{0|1}-\mathcal{A}_{0|1}\mathcal{B}_{0|0}-\mathcal{A}_{0|1}\mathcal{B}_{0|1}$.
Thus we set the symmetrized basis: 
\begin{equation}
\label{eq:chshsymmonomials}
    \widetilde{\mathcal{M}} = \{1, \widetilde{\mathcal{M}}_1 \}. 
\end{equation}
The set of linearly independent elements after averaging is of cardinality 2. 
The moment based relaxation yields $\widetilde{y}_0 = 1$ and $\widetilde{y}_1=\bra{\psi}\widetilde{\mathcal{M}}_1\ket{\psi}$. 
Hence, the moment matrix of the dual depends only on one optimization variable and reads, according to~\eqref{eq:overlineXiZwithredmoments}:
\begin{equation}
    \label{eq: CHSHZoverline}
    \overline{Z} = 
        \begin{pmatrix}
            1 & \frac{1}{2} & \frac{1}{2} & \frac{1}{2} & \frac{1}{2} \\
            \frac{1}{2} & \frac{1}{2} & \frac{1}{4} & \frac{3}{8}-\frac{\widetilde{y}_1}{4} & \frac{1}{8}+\frac{\widetilde{y}_1}{4} \\
            \frac{1}{2} & \frac{1}{4} & \frac{1}{2} & \frac{3}{8}-\frac{\widetilde{y}_1}{4} & \frac{3}{8}-\frac{\widetilde{y}_1}{4} \\
            \frac{1}{2} & \frac{3}{8}-\frac{\widetilde{y}_1}{4}  & \frac{3}{8}-\frac{\widetilde{y}_1}{4}  & \frac{1}{2} & \frac{1}{4} \\
            \frac{1}{2} & \frac{1}{8}+\frac{\widetilde{y}_1}{4} & \frac{3}{8}-\frac{\widetilde{y}_1}{4} & \frac{1}{4} & \frac{1}{2}
        \end{pmatrix}. 
\end{equation}

The partially reduced SDP~\eqref{eq:momentsdpsympartial} is given by
\begin{equation}
    \begin{split}
         d^\star=\max_{\widetilde{y}_1}\;\; & \tilde{b}_0 + \tilde{b}_1 \widetilde{y}_1\\
        \text{s.t.} \;\; &\overline{Z}=A'_0+A'_1\widetilde{y}_1\succeq0.
    \end{split}    
    \label{eq:dualCHSHsym}
\end{equation}
where we have $\tilde{b}_0 = 2$ and $\tilde{b}_1=-4$ and it is straightforward to extract $A_0'$ and $A_1'$ from \eqref{eq: CHSHZoverline}.
This problem can be solved analytically and the optimum $d^\star=2\sqrt{2}$ is reached for ${\widetilde{y}_1}^\star = \frac{1}{2}-\frac{1}{\sqrt{2}}$.

\subsubsection*{Block-diagonalization}
\begin{table*}[]
    \centering
    \begin{tabular}{|c|c|r||r|r|r|r|r|r|r|}
    \hline
         Irrep. & Kernel & Dim. & $\mathbb{1}$ & $a^4$ & $\{a^2,a^6\}$ & $\{a,a^7\}$ & $\{a^3,a^5\}$ & $\{x,a^2 x,a^4 x, a^6 x\}$ & $\{ a x, a^3 x, a^5 x, a^7 x\}$\\
         \hline
         $\sigma^{(1)}$ & $G$ & $1$ & $1$ & $1$ & $1$ & $1$ & $1$ & $1$ & $1$ \\
         \hline
         $\sigma^{(2)}$ & $\langle a \rangle$ & $1$ & $1$ & $1$ & $1$ & $1$ & $1$ & $-1$ & $-1$ \\
         \hline
         $\sigma^{(3)}$ & $\langle a^2, x \rangle$ & $1$ & $1$ & $1$ & $1$ & $-1$ & $-1$ & $1$ & $-1$ \\
         \hline
         $\sigma^{(4)}$ & $\langle a^2, a x \rangle$ & $1$ & $1$ & $1$ & $1$ & $-1$ & $-1$ & $-1$ & $1$ \\
         \hline
         $\sigma^{(5)}$ & $\langle \rangle$ & $2$ & $2$ & $-2$ & $0$ & $\sqrt{2}$ & $-\sqrt{2}$ & $0$ & $0$ \\
         \hline
         $\sigma^{(6)}$ & $\langle a^4 \rangle$ & $2$ & $2$ & $2$ & $-2$ & $0$ & $0$ & $0$ & $0$ \\
         \hline
         $\sigma^{(7)}$ & $\langle \rangle$ & $2$ & $2$ & $-2$ & $0$ & $-\sqrt{2}$ & $\sqrt{2}$ & $0$ & $0$ \\
         \hline
    \end{tabular}
    \caption{Character table of the dihedral group $G^{\text{CHSH}}$.
    For each irrep, $\operatorname{ker} \sigma^{(i)} = \{ g \in G : \sigma^{(i)}_g = \mathbb{1} \}$ is the kernel, i.e. the subgroup that acts through $\sigma^{(i)}$ as the identity.
    If the kernel is the trivial group $\langle \rangle$, the representation is faithful.
    The last columns list the seven conjugacy classes of $G^{\text{CHSH}}$, and below the value of the character for this conjugacy class, i.e. the trace of the image for each irrep.}
    \label{tab:CharacterTableCHSH}
\end{table*}
We will use Proposition.~\ref{prop:blockdiag} to blockdiagonalize $\overline{Z}$.  
For that, we need the action of this group on the generating sequence.
In our case, the generating sequence~\eqref{eq: CHSHQ} is
\begin{equation*}
    \vec{\mathcal{Q}} = (1, \mathcal{A}_{0|0},\mathcal{A}_{0|1},\mathcal{B}_{0|0},\mathcal{B}_{0|1})^\top
\end{equation*}
which is exactly the vector $(1, \mathcal{X}_1, \ldots, \mathcal{X}_4)^\top$ used in~\eqref{eq:affinegroupmatrix}.
Thus, we have simply
\begin{equation}
    \rho: G \to \operatorname{GL}(\mathbb{C}, 5), \qquad \rho_g = g\;.
\end{equation}
To decompose this representation, we consider the irreducible representations (irreps) of $G$, which are summarized in Table~\ref{tab:CharacterTableCHSH}.

The group has $7$ irreps, labelled $\sigma^{(1)}$ to $\sigma^{(7)}$; the first four are one-dimensional, and thus the images of $a$, $x$ can be read directly from the character table:
\begin{align}
\label{eq:chshonedimirreps}
    \sigma^{(1)}_a &= 1, &
    \sigma^{(2)}_a &= 1, &
    \sigma^{(3)}_a &= -1, & 
    \sigma^{(4)}_a &= -1, \nonumber \\
    \sigma^{(1)}_x &= 1,  &
     \sigma^{(2)}_x &= -1, &
     \sigma^{(3)}_x &= 1,  &
     \sigma^{(4)}_x &= -1, 
\end{align}
while, for the last three, we have, with $k=1,2,3$:
\begin{equation}
\label{eq:chshtwodimirreps}
    \sigma^{(k+4)}_a = \begin{pmatrix}
     \operatorname{cos} \frac{2 \pi k}{8} & -\operatorname{sin} \frac{2 \pi k}{8} \\
     \operatorname{sin} \frac{2\pi k}{8} & \operatorname{cos} \frac{2 \pi k}{8}
    \end{pmatrix},
    \quad
    \sigma^{(k+4)}_x = \begin{pmatrix}
     1 & 0 \\ 0 & -1
    \end{pmatrix}.
\end{equation}

Our representation $\rho$ has a decomposition into irreducible representations~\eqref{eq:rhogirrep}:
\begin{equation}
    I^{-1} \rho_g I = \sigma^{(1)}_g \oplus \sigma^{(5)}_g \oplus \sigma^{(7)}_g
\end{equation}
where only three irreducible representations appear with multiplicity one. 

\begin{table}[]
\begin{tabular}{|l|l|l|l|l|l|l|l|}
\hline
$i$ & 1 & 2 & 3 & 4 & 5 & 6 & 7  \\ \hline \hline
$d_i$ & 1 & 1 & 1 & 1 & 2 & 2 & 2  \\ \hline
$m_i$ & 1 & 0 & 0 & 0 & 1 & 0 & 1  \\ \hline
rank($\widetilde{Z}^{(i)}$) & 1 & - & - & - & 1 & - & 0  \\ \hline
rank($\widetilde{X}^{(i)}$) & 0 & - & 0 & - & 0 & 0 & 1  \\ \hline
\end{tabular}
\caption{Numbers of interest for CGLMP $d=2$. In the first row we give the label of the irreducible representations, in the second row we give the respective dimensions of the irreducible representations while in the third row we have their multiplicities.}
\label{tab: info_d=2}
\end{table}

The change of base matrix is the horizontal concatenation $I = \left( I^{(1)} I^{(5)} I^{(7)} \right)$ with $I^{(1)} = (1, 1/2, 1/2, 1/2, 1/2)^\top$
\begin{equation}
\label{eq:CHSHI}
    I^{(5)} = \frac{1}{4} \begin{pmatrix*}
    0 & 0 \\
    -1 & c_+ \\
    -c_+ & -1 \\
    -c_+ & 1\\
    -1 & -c_+
    \end{pmatrix*}, \quad
    I^{(7)} = \frac{1}{4}
    \begin{pmatrix*}
    0 & 0 \\
    -1 & c_- \\
    c_- & 1 \\
    c_- & -1 \\
    -1 & -c_-
    \end{pmatrix*}
\end{equation}
where $c_\pm=\sqrt{2}\pm1$.
Later on, we will need
\begin{equation}
  I^{-(7)}=\begin{pmatrix} 1 & -1-s & s & s &-1-s \\ 0 & s & 1+s & -1-s& -s \end{pmatrix},
\end{equation}
with $s=1/\sqrt{2}$.

After the change of basis, the irreducible representations $\sigma^{(1)}$, $\sigma^{(5)}$ and $\sigma^{(7)}$ are unitary, and therefore  $C_1 = 1$, $C_5=C_7 =\mathbb{1}_2$.
The matrix $\widetilde{Z}^{\text{full}} = I^{-1} \overline{Z}I^{-\dagger}$ is blockdiagonal and has three blocks
\begin{equation}
\label{eq: CHSHZfull}
    \widetilde{Z}^{\text{full}} = 
    \widetilde{A}_0^{(1)} \oplus [(\widetilde{A}_0^{(5)}+\widetilde{A}_1^{(5)}\widetilde{y}_1)\otimes \mathbb{1}_2] \oplus [(\widetilde{A}_0^{(7)}+\widetilde{A}_1^{(7)}\widetilde{y}_1)\otimes\mathbb{1}_2]
\end{equation}
where 
\begin{align}
\label{eq:matAsymCHSH}
\widetilde{A}_0^{(1)} & = 1, &\widetilde{A}_0^{(5)} &= 1/2, &\widetilde{A}_0^{(7)} & = 1/2, \nonumber\\
\widetilde{A}_1^{(1)} & = 0, &\widetilde{A}_1^{(5)} &= -c_-, &\widetilde{A}_1^{(7)} &= c_+.
\end{align}

Hence, the semidefinite positive constraint of Eq.~\eqref{eq:dualCHSHsym} is replaced by three semidefinite positive constraints, i.e. one on each block.
The dual SDP~\eqref{eq:dualCHSHsym} after the blockdiagonalization and using Proposition~\ref{prop: sdpblocs} is given by 
\begin{subequations}
\label{eq: CHSHdual}
    \begin{align}
        \max_{\widetilde{y}_1} \;\; &2-4\widetilde{y}_1\\
        \text{s.t. }\;\; &\widetilde{Z}^{(1)}= 1 \geq 0 \label{eq: CHSHdualCons0}\\
        &\widetilde{Z}^{(5)}= 1/2 + c_- \tilde{y}_1 \geq 0 \label{eq: CHSHdualCons1}\\
        & \widetilde{Z}^{(7)}= 1/2 + c_+ \tilde{y}_1 \geq 0 \label{eq: CHSHdualCons2}.
    \end{align}
\end{subequations}

It is straightforward to solve this LP. 
The objective function of the LP is maximized if the second constraint \eqref{eq: CHSHdualCons1} is saturated, which leads to $\widetilde{y}_1^\star=\frac{1}{2}-\frac{1}{\sqrt{2}}$. 
We display the rank of the moment matrix blocks in Table~\ref{tab: info_d=2}.

\subsection{Symmetries of sum-of-squares certificates}
We now focus on the sum-of-squares certificates studied in Section~\ref{sec:sumofsquares}.
Note that a sum-of-squares certificate stays valid under the action of $G$ because of the group definition~\eqref{eq:defsymgroup}. Revisiting~\eqref{eq:Xsdp} with~\eqref{eq:gXi}, for all $g \in G$, we have
\begin{equation}
\label{eq:symmetrizingformalsos}
    \mathcal{F}  =  g(\mathcal{F}) 
     =  \operatorname{tr}[g(\Xi) X] = \operatorname{tr}[\rho_g^{-1} \Xi \rho_g^{-\dag} X]
     =  \operatorname{tr}[\Xi\rho_g^{-\dag} X \rho_g^{-1}].
\end{equation}

Thus any feasible (respectively optimal) $X$ can be replaced by
\begin{equation}
\overline{X} = \frac{1}{|G|} \sum_{g \in G} \rho_g^\dag X \rho_g.
\end{equation}
which will be feasible (respectively optimal) as well.
We have
\begin{equation}
    \overline{X} = \rho_g^\dag \overline{X} \rho_g
\end{equation}
and thus by Proposition~\ref{prop:blockdiag}
\begin{equation}
\label{eq:blkdiagX}
    \overline{X} = I^{-\dag} \left [ \bigoplus_{i=1}^N \frac{1}{d_i} \widetilde{X}^{(i)} \otimes C_i \right ] I^{-1}.
\end{equation}
Note the normalization factor $\frac{1}{d_i}$ added for later convenience.
We could totally insert this variable $\overline{X}$ in the semidefinite program~\eqref{eq:sossdp}; however, we can reduce the number of equality constraints as well.

Looking at the last equality of~\eqref{eq:symmetrizingformalsos}, we may symmetrize $\Xi$ as well without loss of generality, and write
$\mathcal{F} = \operatorname{tr}[\overline{\Xi} \overline{X}]$ using~\eqref{eq:symXi}. With Proposition~\ref{prop:blockdiag} again:
\begin{equation}
\widetilde{\Xi} := I^{-1} \overline{\Xi} I^{-\dag} = \sum_{\ell=0}^{\tilde{m}} \bigoplus_{i=1}^N \left( \widetilde{A}^{(i)}_\ell \otimes C_i^{-1} \right) \widetilde{\mathcal{M}}_\ell,
\end{equation}
we write, defining a new variable $\widetilde{X} = \bigoplus_{i=1}^{N} \widetilde{X}^{(i)}$:
\begin{eqnarray}
\label{eq:blockdiagprimal}
\mathcal{F} &=&\operatorname{tr}[\overline{\Xi} \overline{X}] = \operatorname{tr}[I^\dag \widetilde{X} I I^{-1} \widetilde{\Xi} I^{-\dag}] = \operatorname{tr}[\widetilde{X} \widetilde{\Xi}] \nonumber \\
&=& \sum_{\ell=0}^{\tilde{m}} \operatorname{tr}\left[\bigoplus_{i=1}^N \left( \frac{1}{d_i} (\widetilde{A}^{(i)}_\ell \widetilde{X}^{(i)}) \otimes (C_i^{-1} C_i) \right) \right ] \widetilde{\mathcal{M}}_\ell \nonumber \\
&=& \sum_{\ell=0}^{\tilde{m}} \operatorname{tr}\left[ \widetilde{A}_\ell \widetilde{X} \right ] \widetilde{\mathcal{M}}_\ell \nonumber \\
&=& \mu - \mathcal{E} = \mu - \tilde{b}_0 - \sum_{\ell=1}^{\tilde{m}} \tilde{b}_\ell \widetilde{\mathcal{M}}_\ell,
\end{eqnarray}

Equating the last two lines of Eq.~\eqref{eq:blockdiagprimal}, we get the semidefinite program:
\begin{equation}
    \begin{split}
    \label{eq:sossdpinvariant}
    p^\star = \min_{\widetilde{X}} \;\; &\tilde{b}_0 + \operatorname{tr} [\widetilde{A}_0 \widetilde{X}]  \\
    \text{s.t.} \;\; &\operatorname{tr} [\widetilde{A}_\ell \widetilde{X}]  =  - \tilde{b}_\ell,\forall \ell\\
     &\widetilde{X}  \succeq  0,
    \end{split}
\end{equation}
which is dual to~\eqref{eq:momentsdpinvariant}.

\subsection{Example: CHSH}
In the CHSH example, the semidefinite variable $X$ has dimension $5 \times 5$, and each optimal solution has a corresponding orbit of optimal solutions $\rho_g^\dag X \rho_g$.
From~\eqref{eq:blkdiagX} with the change of basis matrix~\eqref{eq:CHSHI} for $\rho_g$, we get
\begin{equation}
    I^\dag \overline{X} I = \widetilde{X}^{(1)} \oplus \frac{\widetilde{X}^{(5)} \otimes \mathbb{1}_2}{5} \oplus \frac{\widetilde{X}^{(7)} \otimes \mathbb{1}_2}{7}\;,
\end{equation}
with $\widetilde{X}^{(1)},\widetilde{X}^{(5)},\widetilde{X}^{(7)}\in\mathbb{R}$.
The $X^\star$ from Eq.~\eqref{eq:primalsolutionchsh} is invariant under $G$, and corresponds to the block-diagonal decomposition
\begin{equation}
\label{eq:solutionchshsymprimal}
    \widetilde{X}^{(1)} = \widetilde{X}^{(5)} = 0, \qquad \widetilde{X}^{(7)} = 4 c_- .
\end{equation}

If we do not know the solution in advance, we prefer to solve the semidefinite program~\eqref{eq:sossdpinvariant} where the constant blocks $\mathbb{1}_2$ are eliminated and the constraints simplified.
For that purpose, we reuse the basis~\eqref{eq:chshsymmonomials} leading to the matrices~\eqref{eq:matAsymCHSH}:
\begin{equation}
\label{eq: CHSHprimalblkdiag}
    \begin{split}
         \min_{\widetilde{X}^{(1)}, \widetilde{X}^{(5)}, \widetilde{X}^{(7)}   \in  \mathbb{R} }\;\; &2 + \widetilde{X}^{(1)} + \frac{1}{2} \widetilde{X}^{(5)} + \frac{1}{2} \widetilde{X}^{(7)} \\
            \text{s.t.}\;\; &c_- \widetilde{X}^{(5)} + c_+ \widetilde{X}^{(7)}  = 4\\
            &\widetilde{X}^{(1)}, \widetilde{X}^{(5)}, \widetilde{X}^{(7)}  \ge  0, 
    \end{split}
\end{equation}
for which we get the optimal solution~\eqref{eq:solutionchshsymprimal} again.

To conclude this example, we discuss a way to simplify the construction of certificates.
Let us assume that we can solve the moment-based problem easily, for example by computing the moments of an explicit quantum realization believed to be optimal.
In that case, one can use complementarity~\eqref{eq: complementarity} to reduce the number of variables in the sum-of-squares relaxation.
In our block-diagonal form, it states
\begin{equation}
\label{eq: complementarityblockdiag}
    \widetilde{X}^{(i)} \widetilde{Z}^{(i)} = 0, \qquad \forall i\;.
\end{equation}
Then, one solves the simplified problem and verifies at the end if strong duality holds.
For our problem, the moment-based relaxation~\eqref{eq: CHSHdual} has only the third constraint saturated, see Table~\ref{tab: info_d=2}.
We can thus set $\widetilde{X}^{(1)} = \widetilde{X}^{(5)} = 0$, and obtain the problem
\begin{equation}
    \label{eq: CHSHprimalsym}
    \begin{split}
    \min_{\widetilde{X}^{(7)} \in \mathbb{R}}\;\;\; & 2 + \frac{1}{2} \widetilde{X}^{(7)} \\
     \text{s.t.}\;\; &c_+ \widetilde{X}^{(7)}  = 4\\
      &\widetilde{X}^{(7)} \ge 0,
    \end{split}
\end{equation}
which can be solved directly to obtain again~\eqref{eq:solutionchshsymprimal}.

\subsection{Recovering SOS certificates}

To recover a certificate from the solution of~\eqref{eq:sossdpinvariant}, we compute the LDL decomposition of both the solution and the $C_i$ blocks from Proposition~\ref{prop:blockdiag}:
\begin{equation}
    \widetilde{X}^{(i)} = \widetilde{L}_i \widetilde{D}_i \widetilde{L}_i^\dag, \quad C_i = M_i E_i M_i^\dag\;.
\end{equation}
Note that the computational effort of computing this decomposition is much reduced (Cholesky runs in $\mathcal{O}(n^3)$ if the matrix is $n \times n$).
Now, one has:
\begin{equation}
    \label{eq: SOSfromXfull}
    \mathcal{F} = \vec{\mathcal{Q}}^\dag \overline{X} \vec{\mathcal{Q}} = \vec{\mathcal{Q}}^\dag I^{-\dag} \widetilde{X}^{\text{full}} I^{-1} \vec{\mathcal{Q}}
\end{equation}
and
\begin{equation}
    \widetilde{X}^{\text{full}} = \bigoplus_{i=1}^N \frac{1}{d_i} (\widetilde{L}_i \otimes M_i) (\widetilde{D}_i \otimes E_i) (\widetilde{L}_i \otimes M_i)^\dag
\end{equation}
where the certificate components are given by:
\begin{equation}
    \vec{\mathcal{T}} = \left[ \bigoplus_{i=1}^N \left (\sqrt{\frac{1}{d_i} \widetilde{D}_i} \otimes \sqrt{E_i} \right) \left (\widetilde{L}_i \otimes M_i \right )^\dag \right ] I^{-1} \vec{\mathcal{Q}}.
  \end{equation}

\subsection{Example: CHSH}
Let us recover the SOS from the solution of \eqref{eq: CHSHprimalblkdiag}. 
The only non-zero block is $\widetilde{X}^7 = \widetilde{L}_7 \widetilde{D}_7 \widetilde{L}_7^\dag$ where we set $\widetilde{L}_7 = \sqrt{2}$ and $\widetilde{D}_7 = 2 c_-$.
The factor $C_7 = \mathbb{1}_2$ is the identity and thus we take $M_7 = E_7 = \mathbb{1}_2$ as well.

Using Eq.~\eqref{eq: SOSfromXfull} we get
\begin{equation}
    \begin{split}
    2\sqrt{2}-\mathcal{E}_2 = c_- \Big[ &
    \left \|
    -c_+ \mathcal{A}_{0|0}
    + \mathcal{A}_{0|1}
    + \mathcal{B}_{0|0}
    -c_+ \mathcal{B}_{0|1} + \sqrt{2}
    \right \|^2\\
    + &\left \|
    \mathcal{A}_{0|0}
    +c_+ \mathcal{A}_{0|1}
    -c_+ \mathcal{B}_{0|0}
    - \mathcal{B}_{0|1}
    \right \|^2\Big],
   \end{split}
   \label{eq: CHSHSOSblockdiag}
\end{equation}
where we remind that $c_{\pm} =\sqrt{2}\pm1$, and define $\| \mathcal{S} \|^2 = S^\dag S$.

\section{CGLMP: Generalities}
\label{sec:cglmpgeneralities}
We consider the standard scenario of the CGLMP inequality which consists of two spacelike separated parties with two inputs and $d$ outputs~\cite{Collins_2002}. 
The CGLMP inequality is given by 
\begin{equation}
    \mathcal{E}_d = \sum_{k=0}^{\lfloor{d/2}\rfloor-1}\left(1-\frac{2k}{d-1}\right)(P_k-Q_k),
\end{equation}
where 
\begin{align*}
P_k & =\sum_{i=0}^{m-1} (P(\mathcal{A}_i=\mathcal{B}_i+k)+P(\mathcal{B}_i=\mathcal{A}_{i+1}+k)), \nonumber \\
Q_k & =\sum_{i=0}^{m-1} (P(\mathcal{A}_i=\mathcal{B}_i-k-1)+P(\mathcal{B}_i=\mathcal{A}_{i+1}-k-1)),    
\end{align*}
with
\begin{equation*}
    P(\mathcal{A}_a=\mathcal{B}_b+k) = \sum_{j=0}^{d-1}P(\mathcal{A}_a = j+k \mod d,\mathcal{B}_b=j)\;.
\end{equation*}

\subsubsection*{Polynomial ring}
For the CGLMP scenario with $d$ outputs, the free polynomial ring is given by 
\begin{equation}
\label{eq:PCGLMP}
\begin{split}
    \mathcal{P}^{\text{CGLMP,d}}_{\text{free}} = \mathbb{C}[ \mathcal{A}_{0|0},...,\mathcal{A}_{d-1|0},\mathcal{A}_{0|1},...,\mathcal{A}_{d-1|1},\\\mathcal{B}_{0|0},...,\mathcal{B}_{d-1|0},\mathcal{B}_{0|1},...,\mathcal{B}_{d-1|1}].
    \end{split}
\end{equation}
The quotient ring $\mathcal{P}^{\text{CGLMP},d}$ is defined using the ideal $\mathcal{I}^{\text{Bell}}$~\eqref{eq:ideal}.

\subsubsection*{Symmetries}
\begin{figure}
    \centering
    \includegraphics{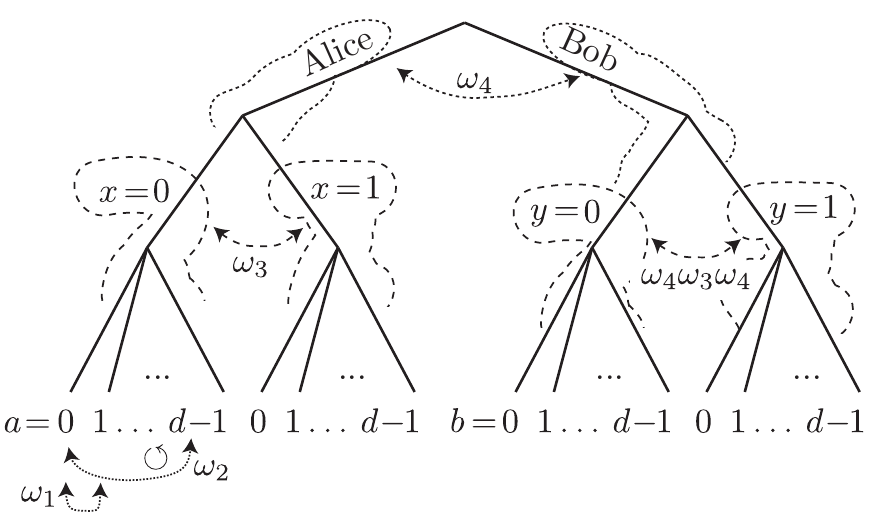}
    \caption{The tree represents the CGLMP scenario with two parties, two measurement settings per party, and $d$ measurement outcomes per measurement. The symmetry group is the automorphism group of the tree, which is generated by $\omega_1$, $\omega_2$ (permutation of outcomes), $\omega_3$ (permutation of Alice's settings), and $\omega_4$ (permutation of parties). For example, the permutation of Bob's settings is given by $\omega_4 \omega_3 \omega_4$.}
    \label{fig:tree}
\end{figure}
The automorphism group of this polynomial ring corresponds has a visual description given in Figure~\ref{fig:tree}.
This group is written using the notation of~\eqref{eq:affinegroupmatrix} and is generated, for $d \ge 3$, by the matrices
\begin{equation}
    \begin{split}
    \label{eq:cglmp:automorphismgenerators}
        \omega_1 &= 1 \oplus \sigma_1 \oplus \mathbb{1}_{4d-6},\\
        \omega_2 & = M_{d,\text{cyclic}} \oplus \mathbb{1}_{3(d-1)}\\
        \omega_3 &= 1 \oplus (\sigma_1 \otimes \mathbb{1}_{d-1}) \oplus \mathbb{1}_{2(d-1)},\\
        \omega_4 &= 1 \oplus (\sigma_1 \otimes \mathbb{1}_{2(d-1)}),
    \end{split}
\end{equation}
where
\begin{equation}
    M_{d,\text{cyclic}} =\left( 
        \begin{array}{rrrrr}
        1 & 0 & \cdots & 0 & 0 \\
        0 & 0 &  &  &  \\
        \vdots & \vdots &  & \mathbb{1}_{d-2} &  \\
        0 & 0 &  &  &  \\
        1 & -1 & \cdots & -1 & -1
        \end{array} \right).
\end{equation}
For $d=2$, we use the same $\omega_3$ and $\omega_4$ and set
\begin{equation}
    \omega_1 = \omega_2 = \begin{pmatrix*}[r]
    1 & 0 \\
    1 & -1
    \end{pmatrix*} \oplus \mathbb{1}_3.
\end{equation}

The matrix $\omega_4$ permutes the parties.
The matrix $\omega_3$ permutes the measurement settings of Alice; by conjugating this element as in $\omega_4 \omega_3 \omega_4$ one permutes the measurement settings of Bob instead.
Finally, the matrices $\omega_1$ and $\omega_2$ permute the outcomes of Alice's first measurement setting.
In particular, $\omega_1$ swaps the first two outcomes, while $\omega_2$ is a cyclic permutation of all outcomes. The two matrices $\omega_1$ and $\omega_2$ generate the symmetric of degree $d$ and order $d!$.
The four matrices generate the group $\Omega$ of order $|\Omega| = 8 (d!)^4$.

We now consider the group leaving the inequality $\mathcal{E}_d$ invariant.
For $d=2,3,4,5$, we verified that this group corresponds to the dihedral group of order $8d$, with presentation
\begin{equation}
    \label{eq:Gcglmp}
   G = \left \langle a, x \middle | a^{4d} = x^2 = x a x a = \mathbb{1} \right \rangle\;.
\end{equation}
For $d=2,3,4,5$, we use
\begin{equation}
    a = \omega_2 \omega_3 \omega_4\;,
\end{equation}
while the form of $x$ depends on the dimension (among the possible generators, we picked the ones with the shortest factorization in $\omega$).

\begin{itemize}
    \item For $d=2$, we had $x_2 = \omega_3 \omega_4 \omega_3$ in~\eqref{eq:CHSHgenerators}.
    
    \item For $d=3$, we write $x_3 = \omega_1 \omega_3 \omega_1 \omega_4 \omega_1 \omega_3 \omega_1$.
    
    \item For $d=4$, we write $x_4 = \omega_3$ $\omega_1$ $\omega_2$ $\omega_1$ $\omega_2^{-1}$ $\omega_1$ $\omega_2$ $\omega_3^{-1}$ $\omega_1$ $\omega_2$ $\omega_1$ $\omega_2^{-1}$ $\omega_1^{-1}$ $\omega_2$ $\omega_3$ $\omega_4^{-1}$ $\omega_1$ $\omega_2$ $\omega_1$ $\omega_2^{-1}$ $\omega_1^{-1}$ $\omega_2$ $\omega_3^{-1}$ $\omega_2^{-1}$ $\omega_1$ $\omega_3$ $\omega_4^{2}$ $\omega_4$ $\omega_3$ $\omega_2$ $\omega_1$ $\omega_2^{-1}$ $\omega_1$ $\omega_2$ $\omega_3$ $\omega_4$.
    
    \item For $d=5$, we write $x_5 = \omega_3$ $\omega_1$ $\omega_2$ $\omega_1$ $\omega_2^{-1}$ $\omega_1$ $\omega_2^{3}$ $\omega_2^{2}$ $\omega_2$ $\omega_1$ $\omega_2^{-2}$ $\omega_2^{-1}$ $\omega_3^{-1}$ $\omega_1$ $\omega_2$ $\omega_1$ $\omega_2^{-1}$ $\omega_1^{-1}$ $\omega_2^{3}$ $\omega_2^{2}$ $\omega_2$ $\omega_1$ $\omega_2^{-2}$ $\omega_2^{-1}$ $\omega_3$ $\omega_4^{-1}$ $\omega_1$ $\omega_2$ $\omega_1$ $\omega_2^{-1}$ $\omega_1^{-1}$ $\omega_2^{3}$ $\omega_2^{2}$ $\omega_2$ $\omega_1$ $\omega_2^{-2}$ $\omega_2^{-1}$ $\omega_3^{-1}$ $\omega_2^{-1}$ $\omega_1^{-1}$ $\omega_2$ $\omega_1$ $\omega_2^{-1}$ $\omega_1$ $\omega_2$ $\omega_3$ $\omega_4^{2}$ $\omega_4$ $\omega_3$ $\omega_2$ $\omega_1$ $\omega_3$ $\omega_4^{2}$ $\omega_4$ $\omega_3$ $\omega_1$ $\omega_2$ $\omega_1^{-1}$ $\omega_2^{-2}$ $\omega_2^{-1}$ $\omega_3^{-1}$ $\omega_4^{-1}$.
\end{itemize}

We did not go beyond $d=5$ for reasons that will be discussed later.

\subsubsection*{Irreducible representations}
The dihedral group of order $8d$ has $2d + 3$ irreducible representations.
The first four were already given in~\eqref{eq:chshonedimirreps}:
\begin{align}
    \label{eq: 1d_irrep}
    \sigma^{(1)}_a &= 1, &
    \sigma^{(2)}_a &= 1, &
    \sigma^{(3)}_a &= -1, & 
    \sigma^{(4)}_a &= -1, \nonumber \\
    \sigma^{(1)}_x &= 1,  &
     \sigma^{(2)}_x &= -1, &
     \sigma^{(3)}_x &= 1,  &
     \sigma^{(4)}_x &= -1, 
\end{align}
and the remaining $2d-1$ are indexed by $k=1,\ldots,2d-1$:
\begin{equation}
    \label{eq: 2d_irrep}
    \sigma^{(k+4)}_a = \begin{pmatrix}
     \operatorname{cos} \frac{2 \pi k}{4d} & -\operatorname{sin} \frac{2 \pi k}{4d} \\
     \operatorname{sin} \frac{2\pi k}{4d} & \operatorname{cos} \frac{2 \pi k}{4d}
    \end{pmatrix},
    \quad
    \sigma^{(k+4)}_x = \begin{pmatrix}
     1 & 0 \\ 0 & -1
    \end{pmatrix}.
\end{equation}

\subsubsection*{Conjectured maximal quantum violations}

In Ref. \cite{Collins_2002} the violation of the inequality for the maximally entangled state of two  dimensional systems is studied. 
There exist a state that better violates the inequality than the maximally entangled state~\cite{Acin_2002,Acin_2005}. 
The measurements maximally violating the CGLMP inequality are conjectured in Ref.~\cite{Collins_2002} to be given by $\mathcal{A}_{a|x}=\ketbra{a}_{a|x}$ and $\mathcal{B}_{b|y} = \ketbra{b}_{b|y}$ with
\begin{align}
\label{eq:cglmpconjecturedmeasurements}
        \ket{a}_{a|x} &= \frac{1}{\sqrt{d}}\sum_{k=0}^{d-1} \exp(\mathbf{i}\frac{2\pi}{d}k(a+\alpha_x))\ket{k}, \nonumber \\
        \ket{b}_{b|y} &= \frac{1}{\sqrt{d}}\sum_{k=0}^{d-1} \exp(\mathbf{i}\frac{2\pi}{d}k(-b+\beta_y))\ket{k},
\end{align}
where the phases are $\alpha_1=1/2$, $\alpha_2=0$, $\beta_1=-1/4$ and $\beta_2=1/4$.
In Ref.~\cite{Zohren2008}, a numerical maximization of the violation of the inequality over all states and measurements shows that the optimal state is not the maximally entangled, but that the best measurements are the same as conjectured best measurements for the maximally entangled states. 
The state maximally violating the inequality is the eigenvector of the Bell operator corresponding to its maximal eigenvalue~\cite{Zohren2008}.

For $d=2$, the conjecture has already been proven. The state maximally violating the inequality is the maximally entangled two-qubit Bell state
\begin{equation}
    \label{eq:cglmpoptimalstate2}
    \ket{\Phi^+}=\frac{1}{\sqrt{2}}(\ket{00}+\ket{11})
\end{equation}
 and the maximal quantum violation is the well-known $\mu_2=2\sqrt{2}$. As already mentioned, we will consider this case only for pedagogical reasons.

For $d=3$, we verify the following conjecture. The state violating the inequality maximally is
\begin{equation}
\label{eq:cglmpoptimalstate3}
     \ket{\psi} = \frac{1}{\sqrt{2+\gamma_3^2}} (\ket{00}+\gamma_3\ket{11} + \ket{22})
 \end{equation}
where $\gamma_3 = \frac{\sqrt{11}-\sqrt{3}}{2}$ and the maximal quantum violation is $\mu_3 = (10-2\gamma_3^2)/3 = 1 + \sqrt{11/3} \approx 2.91485$.

For $d=4$, we verify the following conjecture. The state violating the inequality maximally is
\begin{equation}
\label{eq:cglmpoptimalstate4}
    \ket{\psi} = \frac{1}{\sqrt{2+2\gamma_4^2}}(\ket{00}+\gamma_4\ket{11}+\gamma_4\ket{22}+\ket{33})
\end{equation}
where $\gamma_4\approx 0.7393$ is the $5th$ root of the polynomial $x^8-8 x^7+12 x^6+24 x^5-34 x^4-24 x^3+12 x^2+8 x+1$, when the roots are sorted in increasing order.
The maximal quantum violation is expressed as a polynomial in $\gamma_4$:
\begin{equation}
\mu_4 = \frac{-\gamma_4^7+7\gamma_4^6-4\gamma_4^5-37\gamma_4^4+17\gamma_4^3 + 51 \gamma_4^2 -8\gamma_4 - 9}{3}
\end{equation}
and $\mu_4 \approx 2.9727$.

For $d=5$, the state conjectured to maximally violate the inequality is 
\begin{equation}
    \ket{\psi} = \frac{1}{\sqrt{2+2\gamma_{5_1}^2+\gamma_{5_2}^2}}(\ket{00}+\gamma_{5_1}\ket{11}+\gamma_{5_2}\ket{22}+\gamma_{5_1}\ket{33}+\ket{44}) 
\end{equation}
where $\gamma_{5_1}\approx 0.7189$ is $8th$ root of the polynomial $x^{12}-67x^{10}+294x^8-443x^6+286x^4-75x^2+5$ the and $\gamma_{5_2}\approx 0.6605$ is the $4th$ root of $x^6-14x^5-6x^4+52x^3-4x^2-40x+16$, when the roots are sorted in increasing order. 
The maximal conjectured quantum violation $\mu_5 = 3.0157$ is expressed the $6th$ in increasing order of the polynomial $x^6-65x^4+144x^2+96x+16$.

\section{CGLMP: Computations}
\label{sec:cglmpcomputations}
We already recomputed a known~\cite{Reichardt2013,Bamps_2015} SOS certificate~\eqref{eq: CHSHSOSmomentbasedrelaxation} for the CHSH inequality, which corresponds to the CGLMP inequality with $d=2$, certifying that the maximal quantum bound of this inequality is $2\sqrt{2}$.
For this certificate, the first level of the relaxation was sufficient to provide a tight upper bound, i.e. one that matches the lower bound given by an explicit model.
This level does not provide a tight bound anymore for $d \ge 3$; however, it was conjectured~\cite{Navascues_2008} that the level "$1+AB$" is always sufficient.
This level corresponds to the following generating sequence:
\begin{equation}
\label{eq:cglmpgeneratingsequence}
    \vec{\mathcal{Q}} = \begin{pmatrix} 1 \\ \vec{\mathcal{A}} \\ \vec{\mathcal{B}} \\ \vec{\mathcal{A}} \otimes \vec{\mathcal{B}} \end{pmatrix},
    \vec{\mathcal{A}} = \begin{pmatrix} \mathcal{A}_{0|0} \\ \vdots \\ \mathcal{A}_{d-2|0} \\ \mathcal{A}_{0|1} \\  \vdots \\ \mathcal{A}_{d-2|1} \end{pmatrix}, 
\vec{\mathcal{B}} = \begin{pmatrix} \mathcal{B}_{0|0} \\ \vdots \\ \mathcal{B}_{d-2|0} \\ \mathcal{B}_{0|1} \\ \vdots \\ \mathcal{B}_{d-2|0} \end{pmatrix},
\end{equation}
where all the elements are in the quotient ring defined after~\eqref{eq:PCGLMP}.

\begin{table}
\begin{tabular}{|c|c|c|c|c|c|}
\hline
d & 2 & 2 & 3 & 4 & 5 \\
\hline
level & 1 & 1+AB & 1+AB & 1+AB & 1+AB \\
\hline\hline
$|\vec{\mathcal{Q}}|$ & 5 & 9 & 25 & 49 & 81\\
$|\mathcal{M}|$ & 12 & 25 & 169 & 625 & 1681\\
$|\widetilde{\mathcal{M}}|$ & 2 & 2 & 11 & 26 & 53\\
\hline
\end{tabular}
\caption{Numbers of interest. $|\vec{\mathcal{Q}}|$ is the cardinality of the generating sequence. $|\mathcal{M}|$ and $|\widetilde{\mathcal{M}}|$ are the cardinalities of the linearly independent elements of $\Xi$ and respectively $\overline{\Xi}$. }
\label{tab:dimCGLMP}
\end{table}

In what follows, we start with $d=2$ for pedagogical reasons, and then tackle the cases $d=3$ and $d=4$.
We also give information for the $d=5$ case.
For each $d$, we go through the following steps.

\begin{enumerate}
    \item From the generating sequence $\vec{\mathcal{Q}}$ we form the symbolic matrix $\Xi$ and its average $\overline{\Xi}$.
    We extract a basis $\widetilde{M}$ for the averaged moments $\overline{M}$, see~\eqref{eq:basismoments}.
    The numbers of interest are shown in Table~\ref{tab:dimCGLMP}.
    \item We compute the representation $\rho$ of the group $G$ in~\eqref{eq:Gcglmp} acting on the generating sequence $\vec{\mathcal{Q}}$~\eqref{eq:cglmpgeneratingsequence}. We compute the decomposition of $\rho $~\eqref{eq:rhogirrep} into irreducible representations~\eqref{eq: 1d_irrep} and~\eqref{eq: 2d_irrep} and the change of basis matrix $I$.
    This is done in exact arithmetic.
    The decomposition provides the number and size of blocks in the semidefinite program.       
    \item We calculate the moment matrix $Z$ based on the conjectured optimal measurements~\eqref{eq:cglmpconjecturedmeasurements} and state~\eqref{eq:cglmpoptimalstate2}-\eqref{eq:cglmpoptimalstate4}) in exact arithmetic.
    We compute the block matrices $\tilde{Z}^{(i)}$ and their rank using the change of basis matrix $I$, in exact arithmetic.
    \item 
    Complementarity~\eqref{eq: complementarityblockdiag} gives us restrictions on the block that contribute to the sum-of-squares certificate.
    We observe numerically that the sum-of-squares problem~\eqref{eq:sossdpinvariant} has a family of solutions.
    We perform numerical experiments to set additional blocks $\tilde{X}^{(i)}$ to zero, while keeping the same numerical optimum.
    \item In exact arithmetic, we parameterize the sum-of-squares problem~\eqref{eq:sossdpinvariant} using a computer algebra system.
    Depending on the $d$, the problem is either completely solved, or there remains parameters to fill, that we set so that the SOS block is semidefinite positive. We recover a sum-of-squares certificate.
\end{enumerate}

We note that symmetrization helps doubly in this process.
First of all, it reduces the problem size so that it can be handled in exact arithmetic.
Second, it splits the sum-of-squares problem into blocks, most of whom can be safely neglected.

For computations in exact arithmetic, note that the complexity of the coefficients grows with $d$ even before solving the problem.
Starting with $d = 4$, the optimal state, and thus the coefficients of the moment matrix $Z$ involve a value $\gamma_4$ that cannot be expressed using radicals.
The trick is then to compute in the extension of the rational field generated by $\gamma_4$.
In general, the resulting field also contains the $\cos$ and $\sin$ values necessary to express~\eqref{eq: 2d_irrep} and thus the change of basis matrix $I$.
This is supported efficiently by the \texttt{ToNumberField} command in Mathematica~\cite{Mathematica}.
The resulting field is closed under the four standard operations, and the LDL decomposition~\eqref{eq:chol} can be computed without taking square roots.
Any expression of $\sqrt{D}$ present in the derivations below has to be understood symbolically: we keep the form of $D$ and, at the end, we can extract $(\sqrt{D})^2$ from the final sum-of-squares certificates.

\subsection{$d=2$}
We have already rederived twice a SOS certificate for the CHSH inequality-- or the CGLMP inequality for $d=2$ (Eq.~\eqref{eq: CHSHSOSmomentbasedrelaxation} and Eq.~\eqref{eq: CHSHSOSblockdiag}).

We now solve the problem for the level 1+AB described by the generating sequence~\eqref{eq:cglmpgeneratingsequence}.
Of course, the previous sum-of-squares certificate~\eqref{eq: CHSHSOSblockdiag} is a solution of this larger relaxation as well.
We will not learn anything new; however, this serves as a warmup before tackling larger dimensions.

\begin{table}[]
\begin{tabular}{|l|l|l|l|l|l|l|l|}
\hline
Irrep index $i$ & 1 & 2 & 3 & 4 & 5 & 6 & 7  \\ \hline \hline
Dimension $d_i$ & 1 & 1 & 1 & 1 & 2 & 2 & 2  \\ \hline
Multiplicity $m_i$ & 2 & 0 & 1 & 0 & 1 & 1 & 1  \\ \hline
$\operatorname{rank}(\widetilde{Z}^{(i)})$ & 1 & - & 1 & - & 1 & 0 & 0  \\ \hline
$\operatorname{rank}(\widetilde{X}^{(i)})$ & 0 & - & 0 & - & 0 & 0 & 1  \\ \hline
\end{tabular}
\caption{Dimension of the irreducible representations; the multiplicities provide the size of the semidefinite program blocks. We also show the rank of the primal/dual SDP blocks for CGLMP $d=2$, level $1+AB$. By complementarity, we always have $\operatorname{rank}(\widetilde{Z}^{(i)})+\operatorname{rank}(\widetilde{X}^{(i)})\le m_i$.}
\label{tap: info_d=2_1+AB}
\end{table}

\subsubsection*{Step 1}
We use the generating sequence~\eqref{eq:cglmpgeneratingsequence} which contains $9$ monomials.
There are $25$ linearly independent elements in $\Xi$, hence the cardinality of $\mathcal{M}$ is 25. 
The set of linearly independent elements after averaging over the group $\widetilde{\mathcal{M}}$ is of cardinality 3.
\begin{align}
        \widetilde{\mathcal{M}}_0 =& 1 \nonumber \\ 
        \widetilde{\mathcal{M}}_1 =& \mathcal{A}_{0|1} + \mathcal{B}_{0|0} -
        \mathcal{A}_{0|0}\mathcal{B}_{0|0} + \mathcal{A}_{0|0}\mathcal{B}_{0|1}
         \nonumber \\
        & -\mathcal{A}_{0|1}\mathcal{B}_{0|0} -\mathcal{A}_{0|1}\mathcal{B}_{0|1} \nonumber \\
        \widetilde{\mathcal{M}}_2 =& \mathcal{A}_{0|0}\mathcal{A}_{0|1}\mathcal{B}_{0|0}\mathcal{B}_{0|1}-
        \mathcal{A}_{0|0}\mathcal{A}_{0|1}\mathcal{B}_{0|1}\mathcal{B}_{0|0} \nonumber \\
        & -\mathcal{A}_{0|1}\mathcal{A}_{0|0}\mathcal{B}_{0|0}\mathcal{B}_{0|1}+
        \mathcal{A}_{0|1}\mathcal{A}_{0|0}\mathcal{B}_{0|1}\mathcal{B}_{0|0}
\end{align}
Hence, we have 2 optimization variables (since $\widetilde{y}_0=1$). 

\subsubsection*{Step 2}
The group $G$ leaving the inequality invariant is still $G$, defined in~\eqref{eq:Gchsh} and~\eqref{eq:Gcglmp}.
The representation $\rho: G \to \operatorname{GL}(\mathbb{C}, 9)$ is fully described by the images of the group generators $a$ and $x$:
\begin{equation}
        \rho_a =
         \begin{pmatrix}
         1 & 0 & 0 & 0 & 0 & 0 & 0 & 0 & 0 \\
         1 & 0 & 0 & 0 & -1 & 0 & 0 & 0 & 0 \\
         0 & 0 & 0 & 1 & 0 & 0 & 0 & 0 & 0 \\
         0 & 1 & 0 & 0 & 0 & 0 & 0 & 0 & 0 \\
         0 & 0 & 1 & 0 & 0 & 0 & 0 & 0 & 0 \\
         0 & 1 & 0 & 0 & 0 & 0 & -1 & 0 & 0 \\
         0 & 0 & 1 & 0 & 0 & 0 & 0 & 0 & -1 \\
         0 & 0 & 0 & 0 & 0 & 1 & 0 & 0 & 0 \\
         0 & 0 & 0 & 0 & 0 & 0 & 0 & 1 & 0 
         \end{pmatrix},
\end{equation}
and
\begin{equation}
     \rho_x = 
         \begin{pmatrix}
         1 & 0 & 0 & 0 & 0 & 0 & 0 & 0 & 0 \\
         0 & 0 & 0 & 0 & 1 & 0 & 0 & 0 & 0 \\
         0 & 0 & 0 & 1 & 0 & 0 & 0 & 0 & 0 \\
         0 & 0 & 1 & 0 & 0 & 0 & 0 & 0 & 0 \\
         0 & 1 & 0 & 0 & 0 & 0 & 0 & 0 & 0 \\
         0 & 0 & 0 & 0 & 0 & 0 & 0 & 0 & 1 \\
         0 & 0 & 0 & 0 & 0 & 0 & 1 & 0 & 0 \\
         0 & 0 & 0 & 0 & 0 & 0 & 0 & 1 & 0 \\
         0 & 0 & 0 & 0 & 0 & 1 & 0 & 0 & 0 \\
        \end{pmatrix}.
\end{equation}
The irreducible representations can be read alternatively in~\eqref{eq:chshonedimirreps}-\eqref{eq:chshtwodimirreps} or in \eqref{eq: 1d_irrep}-\eqref{eq: 2d_irrep}.
The representation $\rho$ has a decomposition into five out of the seven irreducible representations
\begin{equation}
    I^{-1} \rho_g I = \left(\mathbb{1}_2 \otimes \sigma_g^{(1)}\right) \oplus \sigma_g^{(3)}\oplus \sigma_g^{(5)}\oplus \sigma_g^{(6)}\oplus \sigma_g^{(7)}
\end{equation}
where the change of basis matrix I computed as explained in Section \ref{subsec: decrep} is given by
\begin{equation}
I=
\begin{pmatrix}
 1 & 0 & 0 & 0 & 0 & 0 & 0 & 0 & 0 \\
 \frac{1}{2} & 0 & 0 & -\frac{1}{4} & \frac{c_+}{4} & 0 & 0 & -\frac{1}{4} & \frac{c_-}{4} \\
 \frac{1}{2} & 0 & 0 & -\frac{c_+}{4} & -\frac{1}{4} & 0 & 0 & \frac{c_-}{4} & \frac{1}{4} \\
 \frac{1}{2} & 0 & 0 & -\frac{c_+}{4} & \frac{1}{4} & 0 & 0 & \frac{c_-}{4} & -\frac{1}{4} \\
 \frac{1}{2} & 0 & 0 & -\frac{1}{4} & -\frac{c_+}{4} & 0 & 0 & -\frac{1}{4} & -\frac{c_-}{4} \\
 \frac{3}{8} & -\frac{1}{4} & \frac{1}{8} & -\frac{d_+}{8} & \frac{d_+}{8} & 0 & -\frac{1}{4} & \frac{d_-}{8} &
   \frac{d_-}{8} \\
 \frac{1}{8} & \frac{1}{4} & \frac{1}{8} & -\frac{1}{4} & 0 & \frac{1}{4} & 0 & -\frac{1}{4} & 0 \\
 \frac{3}{8} & -\frac{1}{4} & -\frac{1}{8} & -\frac{c_+}{4} & 0 & \frac{1}{4} & 0 & \frac{c_-}{4} & 0 \\
 \frac{3}{8} & -\frac{1}{4} & \frac{1}{8} & -\frac{d_+}{8} & -\frac{d_+}{8} & 0 & \frac{1}{4} & \frac{d_-}{8} &
   -\frac{d_-}{8}
\end{pmatrix}
\end{equation}
where $c_\pm = \sqrt{2}\pm1$ and $d_\pm= \sqrt{2}\pm 2$.

This means that the reduced semidefinite program has five blocks, four of dimension $1$ and one of dimension $2$, see Table~\ref{tap: info_d=2_1+AB}.

\subsubsection*{Step 3}
After the blockdiagonalization, $\overline{Z}$ splits into 5 blocks.
The dual SDP reads
\begin{subequations}
    \begin{align}
        \max_{\vec{\widetilde{y}}} \;\; &4-2\widetilde{y}_1\\
        \text{s.t. } 
        &\widetilde{Z}^{(1)} = 
            \begin{pmatrix}
             1 & \widetilde{y}_1 \\ \widetilde{y}_1 & \widetilde{y}_1 + \widetilde{y}_2
            \end{pmatrix}\geq0,  \label{eq: CHSH2dualCons1}\\
        &\widetilde{Z}^{(3)} = 1+4\widetilde{y}_2 \geq 0, \label{eq: CHSH2dualCons2}\\
        &\widetilde{Z}^{(5)} = \frac{1}{2} - c_- \widetilde{y}_1 \geq 0, \label{eq: CHSH2dualCons3}\\
        &\widetilde{Z}^{(6)} =\frac{1}{2}-2\widetilde{y}_2  \geq 0, \label{eq: CHSH2dualCons4}\\
        &\widetilde{Z}^{(7)} =\frac{1}{2} + c_+ \widetilde{y}_1 \geq 0, \label{eq: CHSH2dualCons5}
    \end{align}
    \label{eq: dualCHSH2}
\end{subequations}
and the conjectured optimal measurements correspond to $\tilde{y}_1 = -c_-/2$ and $\tilde{y}_2 = 1/4$, which corresponds to the numerical solution.
The rank of the blocks in this moment-based relaxation are shown in Table~\ref{tap: info_d=2_1+AB}.

\subsubsection*{Step 4}
Complementarity suggests that the variables $\widetilde{X}^{(3)}=\widetilde{X}^{(5)}=0$.
Further numerical investigations enable us to set $\widetilde{X}^{(1)}=\mathbb{0}_2$ and  $\widetilde{X}^{(6)}=0$.
Hence, only $\widetilde{X}^{(7)}$ will matter in the SOS. 
The SOS relaxation turns out to be a linear program and is given by
\begin{equation}
\begin{split}
        \min_{\widetilde{X}^{(7)}}\;\; &2 +\widetilde{X}^{(7)}\\
        \text{s.t.}\;\; &2(\sqrt{2}-1)\widetilde{X}^{(7)} = 4\\
        & \widetilde{X}^{(7)}\geq 0.
\end{split}
\end{equation}
The first constraint fully determines the solution of the LP and $\widetilde{X}^{(7)}= 2(\sqrt{2}-1)$.

\subsubsection*{Step 5}
It is straightforward to extract the SOS certificate
\begin{equation}
    \begin{split}
    \mu_2 - \mathcal{E}_2 = c_- \Big(\|A_{0|0}+c_+ A_{0|1}-c_+
   B_{0|0}-B_{0|1}\|^2+\\
   \|-c_+ A_{0|0}+A_{0|1}+B_{0|0}-c_+B_{0|1}+\sqrt{2}\|^2\Big).
   \end{split}
\end{equation}

\subsection{CGLMP $d=3$} 
The CGLMP inequality for $d=3$ in the Collin-Gisin notation is given by
\begin{equation}
    \begin{aligned}
        \mathcal{E}_3 = &2 - 3(\mathcal{A}_{0|1}+ \mathcal{A}_{1|1}+ \mathcal{B}_{0|0}+\mathcal{B}_{1|0}- \\
        &\mathcal{A}_{0|1} \mathcal{B}_{0|0}- \mathcal{A}_{0|1} \mathcal{B}_{0|1}- \mathcal{A}_{0|1} \mathcal{B}_{1|0}- \mathcal{A}_{0|0} \mathcal{B}_{0|0}- \\
        &\mathcal{A}_{1|0}\mathcal{B}_{0|0}+ \mathcal{A}_{0|0} \mathcal{B}_{0|1}+ \mathcal{A}_{1|0} \mathcal{B}_{0|1}- \mathcal{A}_{1|1} \mathcal{B}_{0|1}-\\
        &\mathcal{A}_{1|0} \mathcal{B}_{1|0}-\mathcal{A}_{1|1} \mathcal{B}_{1|0}+ A_{1|0} \mathcal{B}_{1|1}- A_{1|1} \mathcal{B}_{1|1})
    \end{aligned}
\end{equation}

\subsubsection*{Step 1}
We use the generating sequence~\eqref{eq:cglmpgeneratingsequence} which contains $25$ monomials.
There are $169$ linearly independent elements in $\Xi$, hence the cardinality of $\mathcal{M}$ is 169. 
The set of linearly independent elements after averaging over the group $\widetilde{\mathcal{M}}$ is of cardinality 11. Hence, we have 10 optimization variables (since $\widetilde{y}_0=1$). 

\subsubsection*{Step 2}
We remind that the group G leaving the inequality invariant is given in~\eqref{eq:Gcglmp}.
One can deduce the representations $\rho_g: G \rightarrow \text{GL}(\mathbb{C},25)$ of the group by considering the action of G on the generating sequence $\vec{\mathcal{Q}}$.
We want to decompose $\rho$ into irreducible representation of G.
The group has 9 irreducible representations $\sigma^{(i)}$ with $i=1,...,9$.
Four of the irreducible representations are one dimensional \eqref{eq: 1d_irrep} and five two dimension \eqref{eq: 2d_irrep}.
It turns out that $\rho$ decomposes into two out of the four irreducible representations of dimension 1 with multiplicity 3 and 2 respectively; and the five irreducible representations of dimension 2 with each multiplicity 2. 
The summary of the decomposition can be found in Table~\ref{tab: info_d=3}.
In the Appendix, we give the relevant part of the the change of basis matrix $I$.

\begin{table}[]
\begin{tabular}{|l|l|l|l|l|l|l|l|l|l|}
\hline
$i$ & 1 & 2 & 3 & 4 & 5 & 6 & 7 &8 & 9 \\ \hline \hline
$d_i$ & 1 & 1 & 1 & 1 & 2 & 2 & 2 & 2 & 2 \\ \hline
$m_i$ & 3 & 0 & 2 & 0 & 2 & 2 & 2 & 2 & 2 \\ \hline
rank($\widetilde{Z}^{(i)}$) & 2 & - & 1 & - & 1 & 1 & 0 & 0 & 1 \\ \hline
rank($\widetilde{X}^{(i)}$) & 0 & - & 0 & - & 0 & 0 & 0 & 2 & 1 \\ \hline
\end{tabular}
\caption{Numbers of interest for CGLMP $d=3$, $1+AB$. In the first row we give the label of the irreducible representations, in the second row we give the respective dimensions of the irreducible representations while in the third row we have their multiplicities.}
\label{tab: info_d=3}
\end{table}

\subsubsection*{Step 3 and 4}
Numerical investigations allows us to make the Ansatz that only $\widetilde{X}^{(8)}$ and $\widetilde{X}^{(9)}$ will contribute to the SOS and set the other $\widetilde{X}^{(i)}=0$ with $i=1,...,7$. We use the conjectured states and measurements to construct $Z$ and then blockdiagonalize the matrix using Proposition~\ref{prop:blockdiag} to get $I^{-(8)}ZI^{(8)} = \widetilde{Z}^{(8)} \otimes \mathbb{1}_{n_8}$ and $I^{-(9)}ZI^{(9)} = \widetilde{Z}^{(9)} \otimes \mathbb{1}_{n_9}$.

\subsubsection*{Step 5}
Let us take advantage of the strong duality to parametrize $\widetilde{X}^{(8)}$ and $\widetilde{X}^{(9)}$
\begin{equation}
    \begin{aligned}
        &\widetilde{Z}^{(8)} = \mathbb{0}_2,\;\;\;
        \widetilde{X}^{(8)} = \begin{pmatrix}
         x_1 & x_2 \\
         x_2 & x_3
        \end{pmatrix} \\
        &\widetilde{Z}^{(9)} = \frac{4}{99}\begin{pmatrix}
          c_1 & c_3 \\
         c_3 & c_2
        \end{pmatrix}, \;\;\;
        \widetilde{X}^{(9)} = x_4\begin{pmatrix}
          c_4 & 1 \\
         1 & c_5
        \end{pmatrix}
    \end{aligned}
\end{equation}
with $c_1 = 143 + 77 \sqrt{3} - 39 \sqrt{11} - 21 \sqrt{33}$, $c_2= 110 + 55 \sqrt{3} - 21 \sqrt{11} - 14 \sqrt{33}$, $c_3=121 + 66 \sqrt{3} - 30 \sqrt{11}- 17 \sqrt{33}$, $c_4=\frac{1}{22} (44  - 11\sqrt{3} + 9 \sqrt{11} - 4\sqrt{33})$ and $c_5=\frac{1}{16} (17 - \sqrt{3} - 3 \sqrt{11} + \sqrt{33})$.

It remains to solve the SDP yielding the SOS in its blockdiagonal form 
\begin{equation}
    \begin{aligned}
        \min_{x_1,x_2,x_3,x_4}\;\;\; &2 + \tr[A_0^{(8)}\widetilde{X}^{(8)}]+ \tr[A_0^{(9)}\widetilde{X}^{(9)}]\\
        \text{s.t. } & \tr[A_k^{(8)}\widetilde{X}^{(8)}]+ \tr[A_k^{(9)}\widetilde{X}^{(9)}] = -b_k \;\; k = 1,...,10\\
        &\widetilde{X}^{(8)} \succeq 0, \widetilde{X}^{(9)} \succeq 0
    \end{aligned}    
\end{equation}

At the optimum $\widetilde{X}^{{(8)}^\star}$ and $\widetilde{X}^{{(9)}^\star}$ are respectively given by 
\begin{widetext}
\begin{equation}
    \begin{split}
        &\widetilde{X}^{{(8)}^\star} 
        = \begin{pmatrix}
             \frac{1}{48} \left(15+\sqrt{33}\right) & \frac{3}{32} \left(-5-4 \sqrt{3}+\sqrt{33}\right) \\
 \frac{3}{32} \left(-5-4 \sqrt{3}+\sqrt{33}\right) & -\frac{3}{16} \left(-3-4 \sqrt{3}+\sqrt{33}\right)
        \end{pmatrix}\\
        &\widetilde{X}^{{(9)}^\star} 
        = \begin{pmatrix}
             \frac{1}{16} \left(5 \sqrt{33}-21\right) & -\frac{1}{64} \left(16+\sqrt{3}-3 \sqrt{11}\right) \left(\sqrt{33}-3\right) \\
 -\frac{1}{64} \left(16+\sqrt{3}-3 \sqrt{11}\right) \left(\sqrt{33}-3\right) & \frac{1}{16} \left(3+5 \sqrt{3}\right)
   \left(\sqrt{11}-3\right) \\
        \end{pmatrix}
    \end{split}
\end{equation}

The Choleski decomposition yields

\begin{equation}
    \begin{split}
    L_8 &=  \begin{pmatrix}
                 1 & 0 \\
                 \frac{3}{32}(-27 - 15 \sqrt{3} + 3\sqrt{11} + 5 \sqrt{33}) & 1 \\
            \end{pmatrix}, \\
    D_8 &= \operatorname{diag}\Big(\frac{15 + \sqrt{33}}{48},\frac{3}{128}(-156 - 49\sqrt{3} + 45\sqrt{11}+ 16\sqrt{33})\Big), \\
    L_9 &= \begin{pmatrix}
                 1 & 0 \\
                 \frac{-1632 + 96\sqrt{3} + 288\sqrt{11} - 96\sqrt{33}}{1536} & 1 \\
            \end{pmatrix},\\
    D_9 &= \operatorname{diag}\Big(\frac{5\sqrt{33}-21}{16},0\Big).   
    \end{split}
\end{equation}
\end{widetext}
We know have all elements to write the sum of square SOS 
\begin{equation}
    \begin{split}
   \mu_3-\mathcal{E}_3 =
   &||(\sqrt{D}_8L_8^\dagger\otimes \mathbb{1}_{d_8})I^{(8)} \vec{\mathcal{Q}}||^2 +\\
   &||(\sqrt{D}_9L_9^\dagger\otimes \mathbb{1}_{d_9})I^{(9)} \vec{\mathcal{Q}}||^2.
   \end{split}
\end{equation}

\subsection{CGLMP $d=4$}
For $d=4$, the CGLMP inequality reads
\begin{equation}
    \begin{split}
        \mathcal{E}_4 = &2-\frac{8}{3}(\mathcal{A}_{0|1}+\mathcal{A}_{1|1}+\mathcal{A}_{2|1}+\mathcal{B}_{0|0}+\mathcal{B}_{1|0}+\mathcal{B}_{2|0}-\\
        &\mathcal{A}_{0|1} \mathcal{B}_{0|0}-\mathcal{A}_{0|1} \mathcal{B}_{0|1}-\mathcal{A}_{0|1} \mathcal{B}_{1|0}-\mathcal{A}_{0|1} \mathcal{B}_{2|0}-\\
        &\mathcal{A}_{0|0}\mathcal{B}_{0|0}-\mathcal{A}_{1|0} \mathcal{B}_{0|0}-\mathcal{A}_{2|0} \mathcal{B}_{0|0}+\mathcal{A}_{0|0} \mathcal{B}_{0|1}+\\
        &\mathcal{A}_{1|0}\mathcal{B}_{0|1}-\mathcal{A}_{1|1}\mathcal{B}_{0|1}+\mathcal{A}_{2|0} \mathcal{B}_{0|1}-\mathcal{A}_{2|1} \mathcal{B}_{0|1}-\\
        &\mathcal{A}_{1|0} \mathcal{B}_{1|0}-\mathcal{A}_{1|1} \mathcal{B}_{1|0}-\mathcal{A}_{2|0}\mathcal{B}_{1|0}+\mathcal{A}_{1|0} \mathcal{B}_{1|1}-\\
        &\mathcal{A}_{1|1} \mathcal{B}_{1|1}+\mathcal{A}_{2|0} \mathcal{B}_{1|1}-\mathcal{A}_{2|1} \mathcal{B}_{1|1}-\mathcal{A}_{1|1}\mathcal{B}_{2|0}-\\
        &\mathcal{A}_{2|0} \mathcal{B}_{2|0}-\mathcal{A}_{2|1} \mathcal{B}_{2|0}+\mathcal{A}_{2|0} \mathcal{B}_{2|1}-\mathcal{A}_{2|1}\mathcal{B}_{2|1})
    \end{split}
\end{equation}

\subsubsection*{Step 1}
We use the generating sequence~\eqref{eq:cglmpgeneratingsequence} which contains $49$ monomials.
There are $625$ linearly independent elements in $\Xi$, hence the cardinality of $\mathcal{M}$ is 625. 
There are after averaging over the group 26 linearly independent elements in  $\widetilde{\mathcal{M}}$. Hence, we have 25 optimization variables (since $\widetilde{y}_0=1$). 

\subsubsection*{Step 2}
The group G leaving the inequality invariant is the dihedral group of order 16, $D_{16}$ in~\eqref{eq:Gcglmp}.
The representations of the group $\rho_g: G \rightarrow \text{GL}(\mathbb{C},49)$ are computed by considering the action of G on the generating sequence $\vec{\mathcal{Q}}$.
We want to decompose $\rho$ into irreducible representation of G.
The group has 11 irreducible representations $\sigma^{(i)}$ with $i=1,...,11$. 
Four of the irreducible representations are one dimensional \eqref{eq: 1d_irrep} and seven two dimension \eqref{eq: 2d_irrep}.
It turns out that $\rho$ decomposes into two out of the four irreducible representations of dimension 1 with multiplicity 4 and 3 respectively; and the seven irreducible representations of dimension 2 with each multiplicity 2. 
The summary of the decomposition can be found in Table~\ref{tab:info_d=4}.
The change of basis matrix I is given in the Mathematica notebook.

\subsubsection*{Step 3 and 4}
Also for $d=4$, we used the conjecture to construct $Z$ and then blockdigonalized the matrix to get $Z^{\text{full}}$.\\
Using numerics, we could infer that only three out of the nine blocks of $\widetilde{X}^{\text{full}}$ need to be non-zero.
Complementarity enabled us to reduce the number of free parameters to 12. 

\subsubsection*{Step 5}
For more details about the analytical SOS, we invite the reader to consult the Mathematica file where we provide $I$ and $\widetilde{X}$ \cite{Ioannou2021}.

\begin{table}[]
\begin{tabular}{|l|l|l|l|l|l|l|l|l|l|l|l|}
\hline
$i$ & 1 & 2 & 3 & 4 & 5 & 6 & 7 & 8 & 9 & 10 & 11\\ \hline \hline
$d_i$ & 1 & 1 & 1 & 1 & 2 & 2 & 2 & 2 & 2 & 2 & 2 \\ \hline
$m_i$ & 4 & 0 & 0 & 3 & 3 & 3 & 3 & 3 & 3 &3 & 3\\ \hline
rank($\widetilde{Z}^{(i)}$) & 2 & - & - & 2 & 2 & 1 & 1 & 0 & 0 & 1 & 1 \\ \hline
rank($\widetilde{X}^{(i)}$) & 0 & - & - & 0 & 0 & 0 & 0 & 0 & 3 & 2 & 2  \\ \hline
\end{tabular}
\caption{Numbers of interest for CGLMP $d=4$, $1+AB$. In the first row we give the label of the irreducible representations, in the second row we give the respective dimensions of the irreducible representations while in the third row we have their multiplicities.}
\label{tab:info_d=4}
\end{table}

\subsection{CGLMP $d \ge 5$}

\begin{table}[]
\begin{tabular}{|l|l|l|l|l|l|l|l|l|l|l|l|l|l|}
\hline
$i$ & 1 & 2 & 3 & 4 & 5 & 6 & 7 & 8 & 9 & 10 & 11 & 12 & 13\\ \hline \hline
$d_i$ & 1 & 1 & 1 & 1 & 2 & 2 & 2 & 2 & 2 & 2 & 2 & 2 & 2\\ \hline
$m_i$ & 5 & 0 & 0 & 4 & 4 & 4 & 4 & 4 & 4 & 4 & 4 & 4 & 4\\ \hline
rank($\widetilde{Z}^{(i)}$) & 3 & - & - & 2 & 2 & 2 & 1 & 1 & 0 & 0 & 1 & 1 & 2 \\ \hline
rank($\widetilde{X}^{(i)}$) & 0 & - & - & 0 & 0 & 0 & 0 & 0 & 0 & 3 & 3 & 3 & 2  \\ \hline
\end{tabular}
\caption{Numbers of interest for CGLMP $d=5$, $1+AB$. In the first row we give the label of the irreducible representations, in the second row we give the respective dimensions of the irreducible representations while in the third row we have their multiplicities.}
\label{tab: info_d=4}
\end{table}

We performed steps 1. to 4. numerically for $d=5$ to investigate the scaling of our method, and present the results in Table~\ref{tab: info_d=4}.
We stopped short of computing an exact SOS decomposition.
The conjectured quantum bound and corresponding coefficients of the state were given in Section~\ref{sec:cglmpgeneralities}.
However, we were not able to express them in a common field extension of the rationals: using Mathematica, the polynomial coefficients are then of the order of the million, which renders any further computation cumbersome.

We hoped initially to find family of SOS certificates that would be valid for any $d$.
Without any additional insight on the structure of the problem, this goal is probably out of reach.
Indeed, from our observations, it seems that the multiplicity $m_i$ of the blocks scales linearly with $d$.
This is not surprising, as the relaxation level 1+AB is required to reach optimality.
The length of the generating sequence then scales as $\mathcal{O}(d^2)$, while the group order scales as $\mathcal{O}(d)$.
In comparison, in Ref.~\cite{Tavakoli2019}, the group order grew proportionally with the length of the generating sequence, leading to semidefinite relaxations of constant complexity for any $d$.

\section{Sliwa}
\label{sec:Sliwa}

All facet Bell inequalities have been derived for the scenario consisting of 3 parties with each two measurements and two outcomes by Sliwa~\cite{Sliwa_2003}. 
Sliwa showed that there are 46 inequivalent inequalities. We picked 4 inequalities out of the 46 and computed their sum of square decomposition. 

In Tab.~\ref{tab: Sliwa} one can find for each of the 46 inequalities the symmetry group, the order of the group, the multiplicities $m_i$ and dimensions $d_i$ of the irreducible representations, the cardinality of the generating sequence before and after averaging.
For the generating sequence we used $\vec{\mathcal{Q}}^{\top} = (1,\mathcal{A}_{0|0},\mathcal{A}_{0|1})\otimes(1,\mathcal{B}_{0|0},\mathcal{B}_{0|1})\otimes(1,\mathcal{C}_{0|0},\mathcal{C}_{0|1})$. 
The sequence contains 27 terms.
For some inequalities one gets the same almost quantum bound if one takes only the generating terms with maximum two operators and hence the generating sequence has 19 terms instead of 27. 
To get the SOS, we applied the formalism seen before.
First, we solved analytically the blockdiagonalized moment-based SDP relaxation and then used the complementarity to derive the SOS relaxation.
Note that we did not use any conjectured quantum state and measurements in this case.
Second, we computed analytically the solution of the SOS relaxation and constructed the SOS.
It is straightforward to check the validity of the SOS presented below by expanding the terms.

\begin{widetext}
Sliwa 3rd inequality:
\begin{equation}
    \scalemath{1}{
    \begin{split}
        2\sqrt{2}-\mathcal{E}_{S_3} =
        &\frac{\|2 \sqrt{2} \mathcal{A}_{0|0} \left(1-2 \mathcal{B}_{0|0}\right)+2 \sqrt{2} \mathcal{B}_{0|0}+2 \mathcal{C}_{0|0}+2
   \mathcal{C}_{0|1}-\sqrt{2}-2\|{}^2}{2 \sqrt{2}}+\\
         &\frac{\|2 \sqrt{2} \mathcal{A}_{0|1} \left(1-2 \mathcal{B}_{0|1}\right)+2 \sqrt{2} \mathcal{B}_{0|1}+2 \mathcal{C}_{0|0}-2
   \mathcal{C}_{0|1}-\sqrt{2}\|{}^2}{2 \sqrt{2}}
    \end{split}}
\end{equation}
with 
\begin{equation}
    \begin{split}
       \mathcal{E}_{S_3}= &8 \mathcal{A}_{0|1} \mathcal{B}_{0|1} \mathcal{C}_{0|1}+8 \mathcal{A}_{0|1} \mathcal{B}_{0|1} \mathcal{C}_{0|2}+8 \mathcal{A}_{0|2} \mathcal{B}_{0|2} \mathcal{C}_{0|1}-8 \mathcal{A}_{0|2}\mathcal{B}_{0|2} \mathcal{C}_{0|2}-8 \mathcal{A}_{0|1} \mathcal{B}_{0|1}-\\&4 \mathcal{A}_{0|1} \mathcal{C}_{0|1}-4 \mathcal{A}_{0|1}\mathcal{C}_{0|2}-4 \mathcal{A}_{0|2}\mathcal{C}_{0|1}+4 \mathcal{A}_{0|2} \mathcal{C}_{0|2}-4 \mathcal{B}_{0|1} \mathcal{C}_{0|1}-4 \mathcal{B}_{0|2} \mathcal{C}_{0|1}-4 \mathcal{B}_{0|1}\mathcal{C}_{0|2}+4 \mathcal{B}_{0|2} \mathcal{C}_{0|2}+\\&4 \mathcal{A}_{0|1}+4 \mathcal{B}_{0|1}+4 \mathcal{C}_{0|1}-2
   \end{split}
\end{equation}

Sliwa 10th inequality:
\begin{equation}
    \begin{split}
    4-\mathcal{E}_{S_{10}} = 
        8 \|(\mathcal{C}_{0|0} \left(\mathcal{\mathcal{A}}_{0|0} \mathcal{\mathcal{B}}_{0|0}-1\right)+\mathcal{\mathcal{A}}_{0|1} \left(-\mathcal{\mathcal{B}}_{0|1}
        \left(\mathcal{C}_{0|0}-1\right)+\mathcal{\mathcal{B}}_{0|0} \left(\mathcal{C}_{0|1}-1\right)+\mathcal{C}_{0|0}-\mathcal{C}_{0|1}\right)  \\ +\mathcal{B}_{0|1}
        \left(-\left(\mathcal{A}_{0|0}-1\right) \mathcal{C}_{0|1}+\mathcal{C}_{0|0}-1\right)\|^2
      \end{split}
\end{equation}
with 
\begin{equation}
    \begin{split}
\mathcal{E}_{S_{10}} = &8 A_{0|0} \mathcal{B}_{0|0} \mathcal{C}_{0|0}+8 A_{0|1}\mathcal{B}_{0|0} \mathcal{C}_{0|1}-8 A_{0|1} \mathcal{B}_{0|1} \mathcal{C}_{0|0}-8 A_{0|0}\mathcal{B}_{0|1} \mathcal{C}_{0|1}-\\&8 A_{0|1} \mathcal{B}_{0|0}+8 A_{0|1} \mathcal{B}_{0|1}+8 A_{0|1} \mathcal{C}_{0|0}-8 A_{0|1}\mathcal{C}_{0|1}+8 \mathcal{B}_{0|1} \mathcal{C}_{0|0}+8 \mathcal{B}_{0|1} \mathcal{C}_{0|1}-8 \mathcal{B}_{0|1}-8 \mathcal{C}_{0|0}+4
   \end{split}
\end{equation}

Sliwa 11th inequality:
\begin{equation}
    \scalemath{1}{
    \begin{split}
       4\sqrt{2}-\mathcal{E}_{S_{11}} =  
       &\|\sqrt{\sqrt{2}-1} \left(\mathcal{A}_{0|0}+\mathcal{A}_{0|1}\right)
   \left(\mathcal{B}_{0|0}-C_{0|0}\right)+\sqrt{\sqrt{2}+1} \left(\mathcal{A}_{0|0}-\mathcal{A}_{0|1}\right)
   \left(\mathcal{B}_{0|1}+C_{0|1}-1\right)\|^2+\\
   &\|\sqrt{\sqrt{2}+1} \left(\mathcal{A}_{0|0}+\mathcal{A}_{0|1}\right)
   \left(\mathcal{B}_{0|0}+C_{0|0}-1\right)-\sqrt{\sqrt{2}-1} \left(\mathcal{A}_{0|0}-\mathcal{A}_{0|1}\right)
   \left(\mathcal{B}_{0|1}-C_{0|1}\right)\|^2+\\
   &\|\sqrt{\sqrt{2}+1} \left(\mathcal{A}_{0|0}-\mathcal{A}_{0|1}\right)
   \left(\mathcal{B}_{0|0}+C_{0|0}-1\right)+\sqrt{\sqrt{2}-1} \left(\mathcal{A}_{0|0}+\mathcal{A}_{0|1}-2\right)
   \left(\mathcal{B}_{0|1}-C_{0|1}\right)\|^2+\\
   &\|\sqrt{\sqrt{2}-1} \left(\mathcal{A}_{0|0}-\mathcal{A}_{0|1}\right)
   \left(\mathcal{B}_{0|0}-C_{0|0}\right)-\sqrt{\sqrt{2}+1} \left(\mathcal{A}_{0|0}+\mathcal{A}_{0|1}-2\right)
   \left(\mathcal{B}_{0|1}+C_{0|1}-1\right)\|^2
    \end{split}}
\end{equation}
with
\begin{equation}
    \begin{split}
    \mathcal{E}_{S_{11}} = &-8 \mathcal{A}_{0|0} \mathcal{B}_{0|0} \mathcal{C}_{0|0}-8 \mathcal{A}_{0|1} \mathcal{B}_{0|0} \mathcal{C}_{0|0}-8 \mathcal{A}_{0|0} \mathcal{B}_{0|0} \mathcal{C}_{0|1}+8 \mathcal{A}_{0|1}\mathcal{B}_{0|0} \mathcal{C}_{0|1}+8 \mathcal{A}_{0|0} \mathcal{B}_{0|1} \mathcal{C}_{0|0}-8 \mathcal{A}_{0|1} \mathcal{B}_{0|1} \mathcal{C}_{0|0}+\\&8 \mathcal{A}_{0|0} \mathcal{B}_{0|1}\mathcal{C}_{0|1}+8 \mathcal{A}_{0|1} \mathcal{B}_{0|1} \mathcal{C}_{0|1}+8 \mathcal{A}_{0|0} \mathcal{B}_{0|0}-8 \mathcal{A}_{0|0} \mathcal{B}_{0|1}+8 \mathcal{A}_{0|1}\mathcal{C}_{0|0}-8 \mathcal{A}_{0|1} \mathcal{C}_{0|1}-16 \mathcal{B}_{0|1} \mathcal{C}_{0|1}+8 \mathcal{B}_{0|1}+8 \mathcal{C}_{0|1}-4
   \end{split}
\end{equation}

Sliwa 14th inequality:
\begin{equation}
    \begin{split}
        4\sqrt{2}-\mathcal{E}_{S_{14}} =  &\sqrt{2} \|2 \mathcal{\mathcal{B}}_{0|1} \left(\mathcal{\mathcal{A}}_{0|0}+\mathcal{\mathcal{A}}_{0|1}-\sqrt{2}\mathcal{C}_{0|0}+\sqrt{2}\mathcal{C}_{0|1}-1\right)-2 \mathcal{\mathcal{A}}_{0|1}+\sqrt{2} \mathcal{C}_{0|0}-\sqrt{2} \mathcal{C}_{0|1}+1\|^2+\\
       &\sqrt{2} \|2 \mathcal{\mathcal{B}}_{0|0}-\sqrt{2} \mathcal{C}_{0|0}-\sqrt{2} \mathcal{C}_{0|1}+\sqrt{2}-1\|{}^2
    \end{split}
\end{equation}

\begin{equation}
    \begin{split}
    \mathcal{E}_{S_{14}}= &8 \mathcal{A}_{0|0} \mathcal{B}_{0|1} \mathcal{C}_{0|0}-8 \mathcal{A}_{0|1} \mathcal{B}_{0|1} \mathcal{C}_{0|0}-8 \mathcal{A}_{0|0} \mathcal{B}_{0|1} \mathcal{C}_{0|1}+8 \mathcal{A}_{0|1}\mathcal{B}_{0|1} \mathcal{C}_{0|1}+8 \mathcal{A}_{0|1} \mathcal{C}_{0|0}-8 \mathcal{A}_{0|1} \mathcal{C}_{0|1}+8 \mathcal{B}_{0|0} \mathcal{C}_{0|0}+8 \mathcal{B}_{0|0}\mathcal{C}_{0|1}-\\&8 \mathcal{B}_{0|0}-8 \mathcal{C}_{0|0}+4
    \end{split}
\end{equation}

\end{widetext}

\begin{widetext}
\begin{figure}[h]
\begin{center}
\label{tab: Sliwa}
\begin{tabular}{|l||l|c|l|l|c|c|}
\hline
Sliwa & Group & order & $m_i$ & $d_i$ & $|\mathcal{Q}|$ & $|\widetilde{\mathcal{M}}|$ \\ \hline\hline
1 & $C_2\times S_4$ & 48 & 4  1  2  2  1  3 & 1  1  2  3  3  3 & 27 & 4 \\\hline
2 & $C_4^2 \rtimes D_6$ & 192 & 1  1  2 & 1  6  6 & 19 & 3 \\\hline
3 & $C_2 \wr C_2^2$ & 64 &1  2  1  1  2 & 1  2  2  4  4 & 19 & 2\\\hline
4 & $C_2\times D_8$ & 32 & 4  1  2  2  1  2  1  1  2  2 & 1  1  1  1  2  2  2  2  2  2 & 27 & 6 \\\hline
5 & $D_3$ & 6 & 10 1  8 & 1  1  2 & 27 & 35 \\\hline
6 & $C_2$ & 2 & 15  12 & 1  1 & 27 & 60 \\\hline
7 & $S_4$ & 24 & 1  5  1 & 1  3  3 & 19 & 10 \\\hline 
8 & $C_2^3$ & 8 & 4  2  1  3  3  3  3 & 1  1  1  1  1  1  1  1 & 19 & 16 \\\hline
9 & $C_2^2$ & 4 & 4 7 3 5 & 1 1 1 1 & 19 & 23 \\ \hline
10 & $S_4$ & 24 & 3  4  1  3  1 & 1  3  2  3  1 & 27 & 6\\\hline
11 & $C_2\times D_4$ & 16 & 2  2  1  2  3  3 & 1  1  1  1  2  2 & 19 & 10 \\\hline
12 & $C_2^3$ & 8 & 5  3  5  5  2  3  2  2 & 1  1  1  1  1  1  1  1 & 27 & 21\\\hline
13 & $C_2^3$ & 8 & 2  1  2  2  3  2  4  3 & 1  1  1  1  1  1  1  1 & 19 & 12\\\hline
14 & $C_2^3$ & 8 & 3  2  3  3  2  1  3  2 & 1  1  1  1  1  1  1  1 & 19 & 10\\\hline
15 & $C_2 \times D_4$ & 16 & 3  2  2  1  1  1  1  2  2 & 1  1  1  1  1  1  1  2  2 & 19 & 7\\\hline
16 & $C_2$ & 2 & 18   9 & 1  1 & 27 & 75\\\hline
17 & $D_4$ & 8 & 5  1  2  3  4 & 1  1  1  1  2 & 19 & 13\\\hline
18 & $C_2^2$ & 4 & 10   8   4   5 & 1  1  1  1 & 27 & 39\\\hline
19 & $C_2$ & 2 & 15  12 & 1  1 & 27 & 60 \\\hline
20 & $C_2$ & 4 & 6  4  5  4 & 1  1  1  1 & 19 & 22\\\hline
21 & $C_2$ & 2 & 18 9 & 1  1 & 27 & 75\\\hline
22 & $D_3$ & 6 & 10 1 8 & 1  1  2 & 27 & 35\\\hline
23 & $D_4$ & 8 & 5  2  3  3  7 & 1  1  1  1  2 & 27 & 16\\\hline
24 & $C_2$ & 2 & 18 9 & 1 1 & 27 & 75\\\hline
25 & 1 & -&-&-&-&-\\\hline
26 & $D_6$ & 12 & 6  4  1  4  4 & 1  1  1  2  2 & 27 & 16\\\hline
27 & $C_2$ & 2 & 18   9 & 1  1 & 27 & 75\\\hline
28 & $C_2$ & 2 & 18   9 & 1  1 & 27 & 75\\\hline
29 & 1 &-&-&-&-&-\\\hline
30 & 1 &-&-&-&-&-\\\hline
31 & 1 &-&-&-&-&-\\\hline
32 & 1 &-&-&-&-&-\\\hline
33 & $D_3$ & 6 & 10   1   8 & 1  1  2 & 27 & 35\\\hline
34 & $C_2$ & 2 & 15  12 & 1  1 & 27 & 65\\\hline
35 & $C_2$ &  2 & 15  12 & 1  1 & 27 & 65\\\hline
36 & $C_2$ & 2 & 18   9 & 1  1 & 27 & 75\\\hline
37 & $C_2$ & 2 & 12   7 & 1  1 & 19 & 55\\\hline
38 & $C_2$ & 2 & 11   8 & 1  1 & 19 & 44\\\hline
39 & $D_3$ & 6 & 10   1   8 & 1  1  2 & 2 & 35\\\hline
40 & $C_2$ & 2 & 15  12 & 1  1 & 27 & 60\\\hline
41 & $C_2$ & 2 & 18   9 & 1  1 & 27 & 75\\\hline
42 & $C_2$ & 2 & 18   9 & 1  1 & 27 & 75\\\hline
43 & 1 &-&-&-&-&-\\\hline
44 & $C_2$ & 2 & 10   9 & 1  1 & 19 & 46\\\hline
45 & $C_2$ & 2 & 12   7 & 1  1 & 19 & 55\\\hline
46 & $C_2$ & 2 & 15  12 & 1  1 & 27 & 65\\\hline
\end{tabular}
\end{center}
\caption{For each Sliwa inequality we give the group G and the order of the group. In the fourth and fifth column of the table we give the multiplicities $m_i$ and dimensions $d_i$ of the irreducible representations. $|\mathcal{Q}|$ is the cardinality of the generating sequence while and $|\widetilde{\mathcal{M}}|$ is the cardinality of the basis used after the symmetrization over the group.}
\end{figure}
\end{widetext}

\section{Conclusion}
In this work, we introduced a complete framework for the exploitation of symmetries in noncommutative polynomial optimization problems through the NPA hierarchy.
First, we gave a pedagogical introduction to the NPA hierarchy, with equal attention to the  the extraction of SOS certificates for Bell type inequalities.
In particular, we described in detail the construction of noncommutative polynomial rings, and how the semidefinite relaxations and constructed and solved.
Second, we discussed the symmetries of polynomial rings and of related optimization problems, and applied them to the semidefinite relaxations.
We showed that the corresponding semidefinite programs take a block-diagonal form, and that the number of variables/constraints reduces correspondingly.
Finally, we presented applications.
We investigated the properties of the CGLMP inequality, including its invariance properties and the irreducible representations of its symmetry group.
We proved the optimality of the conjectured optimal quantum measurements for $d=3$ and $d=4$ by constructing SOS certificates.
Finally, we explored Sliwa's inequalities for the Bell scenario containing three parties with binary inputs and outputs.
We provided SOS certificate for the Almost Quantum bound of four inequalities, and listed the symmetry groups of the 46 families present in the scenario.

We focused on computing exact SOS decomposition mostly for pedagogical reasons.
The computation of exact certificates is relevant for the certification of secure protocols, but we think that symmetry reduction will have a greater impact on medium to large scale semidefinite relaxations.
Indeed, most uses of the NPA hierarchy bump into computational time and memory limits.
Currently, the inflation technique~\cite{Wolfe2019b} is the method of choice to investigate the correlations arising from quantum causal structures~\cite{Wolfe2021}.
As the inflation technique is based on a systematic duplication of quantum resources and measurement devices, its formulation has natural symmetries, even in the absence of symmetries in the scenario/correlations under study.
While the exploitation of symmetries reduces the computational complexity, it provides an additional way of selecting elements in the generating sequence.
Indeed, several publications select a subset of the monomials of a given degree~\cite{Navascues2015a,Wolfe2021}; in a different context, the use of machine learning to do so is an active research topic~\cite{Requena2021}.
In the case of CGLMP inequalities, we saw that splitting the generating sequence under symmetry leads to a drastic reduction in the number of elements contributing to a SOS certificate.

Our work proposes a few natural questions that are still open.
First, we wonder whether our framework enables the computation of analytical SOS certificates for other Bell inequalities, and in other tasks such as steering~\cite{Uola2020} and self-testing.
Second, the NPA relaxations of the inflation technique exhibit a high degree of symmetry.
In the context of the Lasserre hierarchy, it was observed~\cite{Riener2012} that the complexity of the hierarchy, after block-diagonalization, did not depend on the number of variables, rather only on the degree of the monomials in the generating sequence.
We conjecture that the same holds for quantum inflation, which would mean that the inflation relaxation hierarchy could be indexed by a single parameter (monomial degree) instead of two (number of copies, monomial degree).
When dealing with large scale problems, another question of practical interest is the scaling of the preprocessing.
In our experience, there are two main bottlenecks in the computation of the block-diagonal form.
The first bottleneck is the computation of linearly independent elements $(\mathcal{M}_k)_k$ and the matrices $(A_k)_k$ defining the semidefinite program.
This bottleneck can be partly alleviated by computing only the relevant parts of the blocks of $\tilde{\Xi}$.
An open question is also to specialize our framework to particular scenarios such as Bell inequalities, as was done for commutative problems in~\cite{Parrilo2003a}.
The second bottleneck is the computation of the change of basis matrix $I$.
If we relax the requirement of exact results and use numerical approximations instead, one can use the approach pioneered by the RepLAB toolbox~\cite{Rosset2021}.
Finally, we mention the exploitation of continuous symmetries.
Our results still hold when the invariance of a problem is described by a compact group rather than a finite group, as the group averaging can be performed by integrating over the Haar measure.
Here again, the black-box computational approach pioneered by RepLAB can handle mixtures of continuous and discrete symmetries.

\medskip

\emph{Acknowledgements.---}We thank Jean-Daniel Bancal, S\'ebastien Designolle, Yeong-Cherng Liang, Alex Pozas-Kerstjens and Elie Wolfe for discussions. We acknowledge financial support from NCCR SwissMAP.
\clearpage

\appendix

\section{Proof of Proposition~\ref{prop:blockdiag}}
\label{app:proof:blockdiag}
Here we prove Proposition~\ref{prop:blockdiag}.

Each $\sigma_g^i$ is similar to a unitary representation $\tau_g^i = (\tau_g^i)^{-\dagger}$, according to a change of basis matrix $T_i$ such that $\sigma_g^i = T_i \tau_g T_i^{-1}$, so that:
\begin{equation}
    \sigma_g = T \tau_g T^{-1}, \; \tau_g = \bigoplus_i \mathbb{1}_{m_i} \! \! \otimes \tau_g^i, \; T = \bigoplus_i \mathbb{1}_{m_1} \!\! \otimes T_i.
\end{equation}

We have
\begin{eqnarray}
    X & = & \rho_g^\dag X \rho_g \nonumber \\
    X & = & I^{-\dag} \sigma_g^\dag I^\dag X I \sigma_g I^{-1} \nonumber \\
    \underbrace{I^\dag X I}_{:= X'} & = & \sigma_g^\dag I^\dag X I \sigma_g \nonumber \\
    X' & = & \sigma_g^\dag X' \sigma_g \nonumber \\
    X' & = & T^{-\dag} \tau_g^\dag T^\dag X' T \tau_g T^{-1} \nonumber \\
    \tau_g T^\dag X' T & = & T^\dag X' T \tau_g
\end{eqnarray}
so the matrix $T^\dagger X' T$ satisfies the assumptions of Schur's lemma, and thus
\begin{eqnarray}
    T^\dagger X' T & = & \bigoplus_i X'_i \otimes \mathbb{1}_{d_i} \nonumber \\
    X' & = & \bigoplus_i X'_i \otimes \underbrace{T_i^{-\dagger} T_i^{-1}}_{= C_i}\;.
\end{eqnarray}
As the change of base matrix $T$ is invertible, all $C_i$ are positive definite.
In particular, their diagonal elements are real and positive.
Finally, if any of the $\sigma^i$ is already unitary, the corresponding $C_i$ is a positive multiple of the identity matrix.

For $Z$ we have
\begin{eqnarray}
    Z & = & \rho_g X \rho_g^\dag \nonumber \\
    Z & = & I \sigma_g I^{-1} Z I^{-\dag} \sigma_g^\dag I^\dag \nonumber \\
    \underbrace{I^{-1} Z I^{-\dag}}_{:= Z'} & = & \sigma_g I^{-1} Z I^{-\dag} \sigma_g^\dag \nonumber \\
    Z' & = & \sigma_g Z' \sigma_g^\dag \nonumber \\
    Z' & = & T \tau_g T^{-1} Z' T^{-\dag} \tau_g^\dag T^\dag \nonumber \\
    T^{-1} Z' T^{-\dag} \tau_g & = & \tau_g T^{-1} Z' T^{-\dag}
\end{eqnarray}
then we again satisfy the assumptions of Schur's lemma and:
\begin{eqnarray}
    T^{-1} Z' T^{-\dag} & = & \bigoplus_i Z'_i \otimes \mathbb{1}_{d_i} \nonumber \\
    Z' & = & \bigoplus_i Z'_i \otimes \underbrace{T_i T_i^\dag}_{= C_i^{-1}}\;.
\end{eqnarray}

\section{CGLMP d=3}
\label{app: cglmp=3}
In this part on the Appendix, we give the relevant part of he change of basis matrix $I$ to resconstruct the SOS for $d=3$.  $\{A_k^{(8)},A_k^{(9)}\}_{k=0,...,10}$.
\begin{equation}
    I^{(8)}= \frac{1}{24}\left(
\begin{array}{cccc}
 0 & 0 & 0 & 0 \\
 4 \sqrt{3}-6 & 2 \sqrt{3} & 3-2 \sqrt{3} & -\sqrt{3} \\
 -2 \sqrt{3} & 6-4 \sqrt{3} & \sqrt{3} & 2 \sqrt{3}-3 \\
 2 \sqrt{3} & 4 \sqrt{3}-6 & -\sqrt{3} & 3-2 \sqrt{3} \\
 2 \sqrt{3}-6 & 6-2 \sqrt{3} & 3-\sqrt{3} & \sqrt{3}-3 \\
 2 \sqrt{3}-6 & 2 \sqrt{3}-6 & 3-\sqrt{3} & 3-\sqrt{3} \\
 2 \sqrt{3} & 6-4 \sqrt{3} & -\sqrt{3} & 2 \sqrt{3}-3 \\
 -2 \sqrt{3} & 4 \sqrt{3}-6 & \sqrt{3} & 3-2 \sqrt{3} \\
 4 \sqrt{3}-6 & -2 \sqrt{3} & 3-2 \sqrt{3} & \sqrt{3} \\
 2 \sqrt{3}-6 & 2 \sqrt{3}-2 & 3 & -\sqrt{3} \\
 2 \sqrt{3} & 2 & -2 \sqrt{3} & -2 \\
 -2 & 2 \sqrt{3} & 2 & -2 \sqrt{3} \\
 4 \sqrt{3}-4 & 0 & -2 \sqrt{3} & 0 \\
 0 & 0 & -\sqrt{3} & -1 \\
 0 & 4-4 \sqrt{3} & 0 & 2 \sqrt{3} \\
 0 & 0 & -2 & 0 \\
 -2 & -2 \sqrt{3} & 2 & 2 \sqrt{3} \\
 2 \sqrt{3}-2 & 2 \sqrt{3}-6 & -\sqrt{3} & 3 \\
 0 & 0 & 2 & 0 \\
 0 & 4 \sqrt{3}-4 & 0 & -2 \sqrt{3} \\
 2 \sqrt{3} & -2 & -2 \sqrt{3} & 2 \\
 -4 & 0 & 4 & 0 \\
 2 \sqrt{3}-2 & 6-2 \sqrt{3} & -\sqrt{3} & -3 \\
 0 & 0 & -\sqrt{3} & 1 \\
 2 \sqrt{3}-6 & 2-2 \sqrt{3} & 3 & \sqrt{3} \\
\end{array}
\right)
\end{equation}

\begin{equation}
    I^{-(8)}=\left(
\begin{array}{cccc}
 1 & 0 & 0 & 0 \\
 0 & 2+\sqrt{3} & 1 & 2+\sqrt{3} \\
 -1-\frac{\sqrt{3}}{2} & \frac{1}{2} & 0 & 0 \\
 \frac{1}{2}+\sqrt{3} & \frac{1}{2} \left(2+\sqrt{3}\right) & \sqrt{3} & 1 \\
 -1 & 1+\sqrt{3} & -1 & \sqrt{3} \\
 -1 & -1-\sqrt{3} & -1 & -\sqrt{3} \\
 \frac{1}{2}+\sqrt{3} & -1-\frac{\sqrt{3}}{2} & \sqrt{3} & -1 \\
 -1-\frac{\sqrt{3}}{2} & -\frac{1}{2} & 0 & 0 \\
 0 & -2-\sqrt{3} & 1 & -2-\sqrt{3} \\
 0 & -1 & 0 & -2 \\
 -\frac{\sqrt{3}}{2} & -\frac{1}{2} & -\sqrt{3} & -1 \\
 -\frac{1}{2} & -\frac{\sqrt{3}}{2} & -1 & -\sqrt{3} \\
 -1 & 0 & -2 & 0 \\
 -\frac{\sqrt{3}}{2} & -\frac{1}{2} & -\sqrt{3} & -1 \\
 -\frac{\sqrt{3}}{2} & \frac{1}{2} & -\sqrt{3} & 1 \\
 -1 & 0 & -2 & 0 \\
 -\frac{1}{2} & \frac{\sqrt{3}}{2} & -1 & \sqrt{3} \\
 \frac{1}{2} & \frac{\sqrt{3}}{2} & 1 & \sqrt{3} \\
 1 & 0 & 2 & 0 \\
 -\frac{\sqrt{3}}{2} & -\frac{1}{2} & -\sqrt{3} & -1 \\
 -\frac{\sqrt{3}}{2} & \frac{1}{2} & -\sqrt{3} & 1 \\
 1 & 0 & 2 & 0 \\
 \frac{1}{2} & -\frac{\sqrt{3}}{2} & 1 & -\sqrt{3} \\
 -\frac{\sqrt{3}}{2} & \frac{1}{2} & -\sqrt{3} & 1 \\
 0 & 1 & 0 & 2 \\
\end{array}
\right)^\top
\end{equation}

\begin{equation}
    I^{(9)}=\frac{1}{24}\left(
\begin{array}{cccc}
 0 & 0 & 0 & 0 \\
 -2 & -2 \sqrt{3} & 3 & 3 \sqrt{3} \\
 -2 & 2 \sqrt{3} & 3 & -3 \sqrt{3} \\
 -2 & 2 \sqrt{3} & 3 & -3 \sqrt{3} \\
 4 & 0 & -6 & 0 \\
 4 & 0 & -6 & 0 \\
 -2 & -2 \sqrt{3} & 3 & 3 \sqrt{3} \\
 -2 & -2 \sqrt{3} & 3 & 3 \sqrt{3} \\
 -2 & 2 \sqrt{3} & 3 & -3 \sqrt{3} \\
 0 & 0 & -1 & \sqrt{3} \\
 -2 & -2 \sqrt{3} & 2 & 2 \sqrt{3} \\
 -2 & -2 \sqrt{3} & 2 & 2 \sqrt{3} \\
 0 & 0 & 2 & 0 \\
 0 & 0 & -1 & -\sqrt{3} \\
 0 & 0 & 2 & 0 \\
 0 & 0 & 2 & 0 \\
 -2 & 2 \sqrt{3} & 2 & -2 \sqrt{3} \\
 0 & 0 & -1 & -\sqrt{3} \\
 0 & 0 & 2 & 0 \\
 0 & 0 & 2 & 0 \\
 -2 & 2 \sqrt{3} & 2 & -2 \sqrt{3} \\
 4 & 0 & -4 & 0 \\
 0 & 0 & -1 & \sqrt{3} \\
 0 & 0 & -1 & \sqrt{3} \\
 0 & 0 & -1 & -\sqrt{3} \\
\end{array}
\right)     
\end{equation}

\begin{equation}
    I^{-(9)}=  \left(
\begin{array}{cccc}
 1 & 0 & 0 & 0 \\
 0 & \sqrt{3} & 1 & \sqrt{3} \\
 -\frac{3}{2} & \frac{\sqrt{3}}{2} & 0 & 0 \\
 -\frac{3}{2} & -\frac{\sqrt{3}}{2} & -1 & -\sqrt{3} \\
 0 & -\sqrt{3} & -1 & -\sqrt{3} \\
 0 & \sqrt{3} & -1 & \sqrt{3} \\
 -\frac{3}{2} & \frac{\sqrt{3}}{2} & -1 & \sqrt{3} \\
 -\frac{3}{2} & -\frac{\sqrt{3}}{2} & 0 & 0 \\
 0 & -\sqrt{3} & 1 & -\sqrt{3} \\
 -3 & 0 & -2 & 0 \\
 -\frac{3}{2} & -\frac{3 \sqrt{3}}{2} & -1 & -\sqrt{3} \\
 \frac{3}{2} & -\frac{3 \sqrt{3}}{2} & 1 & -\sqrt{3} \\
 3 & 0 & 2 & 0 \\
 -\frac{3}{2} & -\frac{3 \sqrt{3}}{2} & -1 & -\sqrt{3} \\
 \frac{3}{2} & -\frac{3 \sqrt{3}}{2} & 1 & -\sqrt{3} \\
 3 & 0 & 2 & 0 \\
 \frac{3}{2} & \frac{3 \sqrt{3}}{2} & 1 & \sqrt{3} \\
 \frac{3}{2} & -\frac{3 \sqrt{3}}{2} & 1 & -\sqrt{3} \\
 3 & 0 & 2 & 0 \\
 \frac{3}{2} & \frac{3 \sqrt{3}}{2} & 1 & \sqrt{3} \\
 -\frac{3}{2} & \frac{3 \sqrt{3}}{2} & -1 & \sqrt{3} \\
 3 & 0 & 2 & 0 \\
 \frac{3}{2} & \frac{3 \sqrt{3}}{2} & 1 & \sqrt{3} \\
 -\frac{3}{2} & \frac{3 \sqrt{3}}{2} & -1 & \sqrt{3} \\
 -3 & 0 & -2 & 0 \\
\end{array}
\right)^\top 
\end{equation}

\bibliography{sos_references} 

\end{document}